\documentclass[journal,12pt, draftclsnofoot, onecolumn]{IEEEtran}
\ifCLASSINFOpdf
   \usepackage[pdftex]{graphicx}
   \graphicspath{{../pdf/}{../jpeg/}}
   \DeclareGraphicsExtensions{.pdf,.jpeg,.png}
\else
   \usepackage[dvips]{graphicx}
   \graphicspath{{../eps/}}
   \DeclareGraphicsExtensions{.eps}
\fi
\usepackage{amsmath,nccmath,amsfonts,amssymb}

\usepackage{algorithm}% http://ctan.org/pkg/algorithm
\usepackage{algpseudocode}% http://ctan.org/pkg/algorithmicx
\usepackage{setspace}
\usepackage{booktabs}
\usepackage{amsthm}
\newtheorem{theorem}{Theorem}[subsubsection]
\newtheorem{lemma}{Lemma}[subsubsection]
\newtheorem{corollary}{Corollary}[subsubsection]
\theoremstyle{definition}
\newtheorem{definition}{Definition}[subsubsection]
\newtheorem{example}{Example}[subsubsection]
\usepackage[normalem]{ulem}
\useunder{\uline}{\ul}{}
\usepackage{booktabs}
\usepackage{xcolor}
\begin{document}
\title{Delay Analysis of Multichannel Parallel Contention Tree Algorithms \\ MP-CTA}

\author{Murat~G\"ursu,~\IEEEmembership{Student~Member,~IEEE}, Alberto~Mart\'inez~Alba,
       %,
	   ~and~Wolfgang~Kellerer,~\IEEEmembership{Senior~Member,~IEEE}}% <-this % stops a space


\markboth{Transactions on Information Theory}%
{Shell \MakeLowercase{\textit{et al.}}: Bare Demo of IEEEtran.cls for IEEE Journals}



\maketitle

\begin{abstract}
%Tree algorithms are initially provided as a trade-off between latency and central information. Partial central information enables lower delays compared to distributed systems. This perspective is leveraged for the massive machine type communication. The analysis for probability density function for number of slots used by a tree resolution and the analysis for number of probability distribution of number of levels for a tree resolution is available. The former can give the delay pdf in case a single channel is used while the latter produces the delay pdf in case non-constraint number of channels are parallelizing the tree resolution. In this work we investigate the delay pdf for tree resolution with particular number of channels or parallelization and provide precise analytical results confirmed via simulations.% if you wish. %This we believe can be used to provide requirement based parallelization of sensor access in massive machine type communication.

Contention tree algorithm is initially invented as a solution to improve the stable throughput problem of Slotted ALOHA in multiple access schemes. Even though the throughput is stabilized in tree algorithms, the delay of requests may grow to infinity with respect to the arrival rate of the system. Delay depends heavily on the exploration of the tree structure, i.e., breadth search, or depth search. Breadth search is necessary for faster exploration of tree. The analytical probability distribution of delay, which is available mostly for depth search, is not generalizable to all breadth search. In this paper we fill this gap through though arbitrary grouping of branches and including this in the delay analysis. This enables obtaining the delay analysis of any contention tree algorithm that runs a breadth first search exploration. We show through simulations that the analysis is in agreement with the realizations.
\end{abstract}


% \fi
\IEEEpeerreviewmaketitle

\section{Introduction}

%\IEEEPARstart{T}{he} multiple dimensions of wireless communications allow for a myriad of different strategies to exploit their possibilities. Since wireless communications are based on handling electromagnetic waves, those dimensions are at least five: time, frequency, energy, space and code. Technology advancements control how granularly we can create resources out of these five dimensions. Once we are able to quantify and combine these dimensions into resource blocks, we can create protocols that use them.

\IEEEPARstart{P}{rotocols} for resource management can be roughly categorized as contention-based, such that users are not assigned resources but they contend for them, and contention-free, where each user has separate access to allocated resources. The suitable protocol is selected depending on the requirements of a system. Contention-free protocols provide guarantees for certain services like industrial control, whereas the contention-based protocols enable flexible use of resources, as highly dynamic requests are required in, e.g., mobile networks. The flexibility of contention-based communication has been attracting recent interest due to the upcoming concept of Internet of Things. The number of users is expected to grow exponentially such that pre-allocated resource management is sub-optimal. %The current contention-based algorithms, however, have limitations that have been thoroughly investigated in previous literature. 
%efficiency loss due to allocating constant resources for any event that may happen is like living in Japan and keeping thousands of people ready to run away from a Tsunami since earthquake may arrive anytime. However, disasters can be resolved once they are detected. Similar perspective is used in communication networks, to diverge from TDMA standards to contention based ones. 

%ALOHA is one of the oldest and simplest contention-based protocols, initially used to communicate between island campuses in Hawaii. The motto of the protocol is `send and pray', that is, data are transmitted whenever available and no special mechanisms for detecting or avoiding collisions are in place. Surprisingly, ALOHA works reasonably well up to certain traffic levels. This protocol is further improved via synchronization, leading to so-called
 Slotted-ALOHA (SA) is one protocol that deals with multiple access without reservations. Nonetheless, problems regarding stability are still present. Tree Algorithms working on top of SA may alleviate these problems. Stabilization for SA is achieved via sending successive feedback to collided users, such that they are prioritized compared to initial arrivals.
% Synchronization implies that data transmission is limited to pre-defined time slots, which decreases the probability of two packets overlapping, thus reducing collisions.
Contention tree algorithms are well known for stable throughput operation. Analysis of throughput of the algorithms is well established in the state of the art while the delay analysis is limited. The distribution of the delay is only available under certain settings. In this paper we generalize this analysis for any breadth-first search for contention tree algorithm in a multichannel environment and call this new approach \textit{Multichannel Parallel Contention Tree Algorithm} (MP-CTA).

The structure of the paper is as follows. In Section~\ref{sec:BACK} we provide details about the state of the art on delay analysis on contention tree algorithms. In Section III we introduce our model and the analysis. In Sec.~\ref{Simusimu} simulations are given to show that analytical assumptions match realizations. In Sec.~\ref{sec:CONC} the key contributions are repeated, which concludes the paper.

\section{Background}
\label{sec:BACK}
%\subsection{System Model}
Contention Tree Algorithms (CTA) are designed to provide efficient medium access to a channel connecting a central station and a set of \textit{contenders}, which are not aware from one another. Contenders are devices that try to send information to the central station through a shared communication channel. 
\subsection{Binary Tree Algorithm (BTA)}
The Binary Tree Algorithm  was invented, by Tsybakov and Mikhailov in 1978 \cite{mikhailov} and independently by Capetanakis in 1979 \cite{capetanakis1979tree}. It is simple to implement and only requires binary feedback from the central station. The principle behind this algorithm is a tree-like splitting strategy. First, contenders access a slot randomly. If multiple contenders access the channel at the same time the result is a \textit{collision}. In such an event, all the initial arrivals are blocked. After a collision, contenders draw a random binary number, either 0 or 1. Those which selected 0 are allowed to transmit in the following slot and those which selected 1 wait until the second slot or until the full resolution of those which selected 0 (depending on the implementation). This random splitting is repeated after every collision until no collisions appear. At that point, it is guaranteed that every device has successfully accessed the channel.

The operation of a BTA can be depicted as a tree diagram, like the one shown in Fig. \ref{fig:BTA_example}. In such a diagram, each group of devices with the same sequence of splittings is represented by a \textit{node}. The number inside each node reflects the number of contenders that have reached that node. In case of collision, that is, if the number of contenders in the node is greater than one, two new branches sprout from the collided node, since the contenders are divided into two new groups. The numbers by each branch represent the two possible choices that a contender can make.
\begin{figure}
\centering
\includegraphics[width=0.4\linewidth]{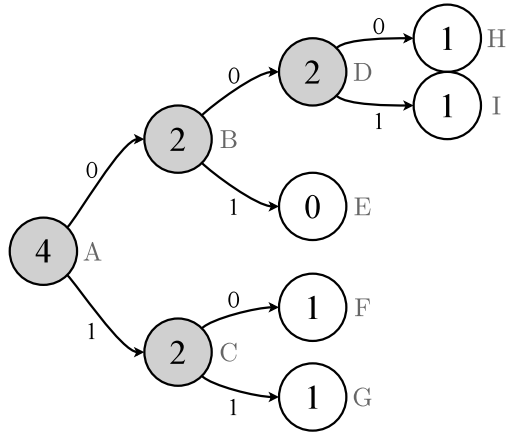}
\caption{Tree representation of an example of a Binary Tree Algorithm, with 4 initial contenders.}
\label{fig:BTA_example}
\end{figure}
\subsection{Preliminaries}
At this point we introduce the terminology for CTA. The initial collision, as also the source of the tree is called the \textit{root}. Each node in the tree except the first one is alternatively called a \textit{contention slot}, instead of simply \textit{node}. The maximum number of branches stemming from a contention slot is called the branching factor and denoted as $Q$.  Immediate children of the same contention slot is called a \textit{contention frame} as a group. A contention frame will contain at maximum $Q$ contention slots. 
%In the example it will contain two slots due to binary structure. The \textit{level} of a contention slot is defined as the depth of the contention slot. So the number of contention slots at level 2 in the example is 4. 

We introduce also a time-slotted collision channel model with immediate perfect feedback, we will refer to this as \textit{channel}. We can have multiple of these channels that are available to use in parallel.

%We want to emphasis here that intuitively PTA is optimal for lower delay in a tree resolution. And we focus on this for our delay analysis.

\subsection{Delay Definition}
Given the preliminaries, we will introduce the delay concept in trees. Delay is the time required to resolve a user, i.e., number of \textit{time slots} from the root of the tree to the successful slot of the contender. The definition of the delay is bound to the number of \textbf{simultaneously usable channels} and \textbf{exploration technique} of the tree.

Here it is important to emphasize that the BTA algorithm is initially designed as a \textit{Serial Tree Resolution} such that a depth first search (DFS) is done in the tree. However, Capetanakis also suggested a breadth first search (BFS) version of the algorithm, and called it \textit{Parallel Tree Resolution} (PTA). 
 
Using the values in Fig.~\ref{fig:BTA_example}, we investigate how different exploration of the delay can affect the delay. In Table~\ref{BFS_tree} we show the evolution of the tree for DFS and BFS, where columns depict evolving time. The contention slots with successes are $( \text{F} , \text{G}, \text{H}, \text{I} )$. So if we write the delay in the same order $D(DFS) = ( 8,9,4,5 )$ and $D(BFS) = ( 6,7,8,9 )$ , where $D(\cdot)$ is the delay function. We get different delay values for each request. %In a single channel system this selection may not make a difference but, in existence of multiple channels BFS is more intuitive to enable lower delays.
\begin{table}
\centering
\begin{tabular}{ c | c | c | c | c | c | c | c | c | c }
	Evolution type & t=1 & t=2 & t=3 & t=4 & t=5 & t=6 & t=7 & t=8 & t=9 \\ \hline
	\textbf{DFS} & A & B & D & H & I & E & C & F & G \\ \hline
	\textbf{BFS} & A & B & C & D & E & F & G & H & I \\ \hline
\end{tabular}
\caption{Breadth First Search and Depth First Search comparison in single channel environment}
\label{BFS_tree}
\end{table}

\subsubsection{Single Channel}
In a \textit{single channel} system the delay maps to the number of nodes (contention slots). For instance, the probability to be successful at $i^\text{th}$ contention slot is an one to one mapping to the delay distribution of the contenders in the resolution. In \cite{molle1993conflict}, the probability generating function for the $i^\text{th}$ successful contention slot conditioned on the initial number of contenders is given such that it can be used to derive the delay distribution. Conditioning on number of initial collided users has been a common approach in most of the work \cite{janssenjong}. We also want to mention that some work focused on Poisson arrivals instead of conditioning on the initial number of contenders \cite{huang}. We think that using the former extends the analysis to be applicable to many arrival distributions and we will also use this approach.

\subsubsection{$Q$ Channel}
In a \textit{$Q$ channel} system, where $Q$ is the number of branches stemming from a contention slot, a full contention frame can be explored parallel at the same time-slot. In this case the delay will map to the probability distribution of success in a contention frame. %Intuitively, selecting one contention slot or another does not change the success probability in one contention frame. However, if a STA is used, the delay analysis becomes complicated. In such a scenario if both of the contention slots resulted in a collision one has to wait for the full resolution of the other slot. Thus, delay analysis has to incorporate that making it complex. In case of PTA it is intuitive though \mg{agree?} that can be derived from distribution of contention slots which converts the 
 Contracting each contention frame to one node keeps the tree structure in tact. Similar recursive analysis to the one used for a success in contention slot is used for contention frames.
%Shifting the perspective from contention slots to contention frames keeps the tree structure in tact. Through modifying the probabilities the same analysis used for successful contention slots can be applied to contention frames. 

In \cite{kaplanMAT} they have conducted such analysis, where they used a recursion instead of the PGF.  

\begin{table}
	\centering
	\begin{tabular}{ c | c | c | c | c | c | c | c | c | c }
		Evolution type & t=1 & t=2 & t=3 & t=4 & t=5 & t=6 & t=7 & t=8 & t=9 \\ \hline
		\textbf{DFS} & A & B & D & H & F &  &  &  &  \\ 
		             &   & C & E & I & G &  &  &  &  \\ \hline
		\textbf{BFS} & A & B & D & F & H &  &  &  &  \\ 
		             &   & C & E & G & I &  &  &  &  \\ \hline
	\end{tabular}
	\caption{Breadth First Search and Depth First Search comparison in $Q$ channel environment}
	\label{BFS_tree_q}
\end{table}

In Tab.~\ref{BFS_tree_q} we have extended the tree evolution to $Q=2$ channels consistent with the tree example. The delays for both cases are $D(DFS) = ( 5,5,4,4 )$ and $D(BFS) = ( 4,4,5,5 )$. %The delays for both cases looks similar in this scenario. 

\subsubsection{Infinite Channel}
We refer the case as infinite channels where the number of channels compared to the number of contention slots in any level of the tree is larger. In such a case all the contention slots in one level of the tree can be explored at the same time-slot. Thus, the probability of success in the $m^{\text{th}}$ level of the tree can be used as delay for resolution of one user in the tree. The probability of success of a contender in level $m$ conditioned on the initial number of contenders for $Q$-ary trees are given in \cite{janssenjong}. They have also derived the probability that the tree terminates at level $M$ conditioned on the initial number of contenders. Such use of channels is not practical since the number of channels required for each time-slot is changing while the tree is evolving, and grows exponentially with respect to $Q$.

 \begin{table}
 	\centering
 	\begin{tabular}{ c | c | c | c | c | c | c | c | c | c }
 		Evolution type & t=1 & t=2 & t=3 & t=4 & t=5 & t=6 & t=7 & t=8 & t=9 \\ \hline
 		\textbf{BFS} & A & B & D & H &  &  &  &  &  \\ 
 		             &   & C & E & I &  &  &  &  &  \\ 
 		             &   &  & F &  &  &  &  &  &  \\ 
 		             &   &  & G &  &  &  &  &  &  \\ \hline
 	\end{tabular}
 	\caption{Breadth First Search and Depth First Search comparison in infinite channel environment}
 	\label{BFS_tree_inf}
 \end{table}
 
 In Tab.~\ref{BFS_tree_inf} we have extended the tree evolution to infinite channels such that each time all the level can be transmitted. The delays are $D(BFS) = ( 3,3,4,4 )$. This analysis also give the minimum delay achievable in a CTA without interrupt, since all the contention slots in a level is transmitted. Exploration of a level before the prior is not possible. There are also tree algorithms that restarts the tree from a specific node depending on the feedback obtained form the channel \cite{fcfs} \cite{popo}. However, this requires all the devices to keep listening during the resolution and thus, is not considered in this work.

\subsection{Arbitrary number of Channels}

The number of channels for the tree resolution can be fixed to an arbitrary value. In this case we define Multichannel Parallel Contention Tree Resolution (MP-CTA). We group $H$ contention slots into one time-slot and schedule time-slots consecutively. The grouping is not done cross levels, such that given $Q^m$ contention slots existing at level $m$, $\lceil \frac{Q^m}{H}\rceil$ time-slots are used to explore that level of the tree before proceeding to the next level.

 The number of simultaneously explorable contention slots increase with levels. Parallelization that is higher than the slots in a level can result in inefficiencies i.e., using $H$ channels waste $H-1$ resources for the initial contention. We define number of arbitrary channels $H$ such that $H= G \cdot Q$. This parallelization up to the level $M$ where number of contention slots is greater than the number of channels e.g., $Q^M \geq H$,  wastes, $M\cdot H - \sum_{m=0}^{M} Q^m$, resources. We define $H$ to be a multiplicative of $Q$, i.e., $  H = G \cdot Q$. In our analysis we restrict ourselves to $Q = 2$.

\begin{figure}
	\centering
	\includegraphics[width=1\linewidth]{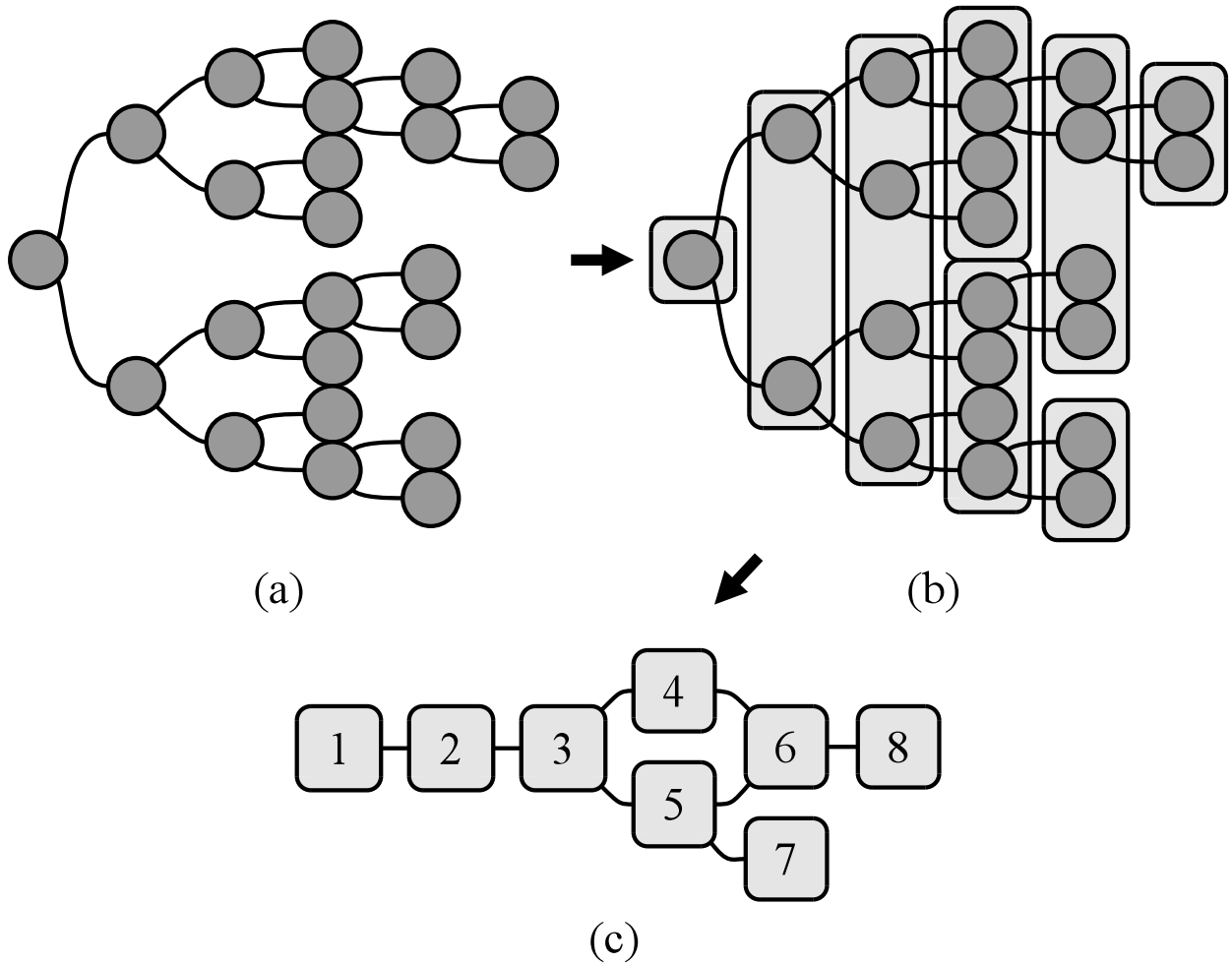}
	\caption{(a) Diagram of a classical Binary Tree Algorithm, (b) Grouping of nodes into slots (depicted as rectangles) in a multichannel tree. (c) Resulting slots and order of transmission.}
	\label{fig:Evolution2MC}
\end{figure}

%Since a new node can only be created after the transmission of its parent nodes, only those nodes that are not descendant from one another can be transmitted at the same time. The simplest way to satisfy this condition is to limit the simultaneous transmission of nodes to those belonging to the same level of the tree, as it is depicted in  Fig. \ref{fig:Evolution2MC}. In this context, a level is defined as the set of nodes that have the same number of ancestors. Henceforth, we will assume that any tree starts at the level 0, which contains the initial collision. This number scheme will be useful in the following analysis.
In Fig. \ref{fig:Evolution2MC} we share an example of MP-CTA. We assume a MP-CTA with $G = 2$ such that the grouping is done for 4 contention slots in  Fig. \ref{fig:Evolution2MC}b. 

We see that the number of contention slots in a certain level can be greater than $H$. In that case, the level needs to be broken into multiple time slots. Based on this fact, we can regard a time-slot as the set of contention slots that are transmitted at the same time. As a consequence, we can group contention slots into time-slots and group time-slots into levels as depicted in Fig. \ref{fig:Evolution2MC}b.

%We model the number of available channels in the system, i.e. the maximum number of simultaneously transmitted nodes, by means of the variable $G$. For the sake of simplicity in the subsequent formulae, $G$ is chosen to denote the maximum number of \textit{pairs} of nodes fitting in one slot, instead of just the number of nodes. For instance, in Fig. \ref{fig:Evolution2MC}b a scenario with four channels is depicted, hence a maximum of four nodes are allowed to lie in the same slot. This implies that $G=2$, since two pairs of nodes can be transmitted simultaneously, i.e. within the same slot. It is important to note that not every slot contains $2G$ nodes, as can be also observed in the figure.

As intuitive more parallelization, we use \textit{breadth-first traversal}, which is presented in Fig. \ref{fig:Evolution2MC}c. The classical way of traversing a Binary Tree Algorithm, which relies on solving the collisions in a nested manner (\textit{depth-first traversal}), is not as suitable since the tree structure of the algorithm can be modified by the parallelization. This loss of the tree structure is also shown in Fig. \ref{fig:Evolution2MC}c.

%Finally, we may proceed to study the performance of multichannel trees with respect to classical single channel binary trees. We aim to show that  multichannel Tree Algorithms outperform single channel Tree Algorithms in terms of time to complete the tree and in terms of the access delay experienced by a contender, as intuitively expected. Of course, the superior performance of multichannel Tree Algorithms is due to the use of additional channels, not to the algorithms themselves. However, the comparison between multichannel and single channel Tree Algorithms is necessary to fully understand the behavior of the former.  

In the following sections, a complete analysis of the statistics of the number of time-slots that are required to complete the tree and the statistics of the access delay experienced by a contender will be derived. Time slot is defined $G$ contention frames or $H$ contention slots grouped to transmit simultaneously.

\section{Multichannel Parallel - Contention Tree Algorithms (MP-CTA)}
\label{sec:PAR}
In the Table \ref{tabletebla} the most relevant variables which were used in this section is presented. Some of them will be reused in computing the access delay.
\begin{table}[]
	\centering
	\caption{Most relevant variables for computing $p_{\mathcal{T}^N}(t)$}
	\label{tabletebla}
	\begin{tabular}{@{}clc@{}}
		\toprule
		\textbf{Variable} & \multicolumn{1}{c}{\textbf{Definition}}  &   \multicolumn{1}{c}{\textbf{Definition index}}                             \\ \midrule
		\vspace*{1mm} 
		$N$               & Number of initial contenders.    &            -                                                                      \\
				\vspace*{1mm}
		$G$              & Number of contention frames per time slot.   &      -                                                                         \\
		\vspace*{1mm}  
		$m$               & Level index.                       &              -                                                                   \\
		\vspace*{1mm}
		$M$               & Number of considered levels.         &             -                                                                  \\
		\vspace*{1mm}
		$\mathcal{T}^N$   & Number of time slots.                  & \ref{{def_T}}                                                       \\
		\vspace*{1mm}
		$\mathcal{T}_m^N$ & Number of time slots at the level $m$.       &   \ref{{def_Tm}}                                                          \\
		\vspace*{1mm}
		$\mathcal{X}_m^N$ & Number of collisions at the level $m$.      &    \ref{{def_Xm}}                                             \\
		\vspace*{1mm}
		$\mathcal{K}_m^N$ & Number of contenders at the level $m$.      &   \ref{{def_Km}}                                              \\
		\vspace*{1mm}
		$\mathcal{Y}_\eta$   & \begin{tabular}[c]{@{}l@{}}Number of children collisions of a \\ parent collision with $\eta$ contenders.\end{tabular} &  \ref{{def_Y}}  \\ \bottomrule
	\end{tabular}
\end{table}

\subsection{Number of required time slots required to complete the tree}
As it was stated before, the grouping of contention frames into time slots erases the recursive properties of contention tree algorithms. As a consequence, recursive approaches to obtain the length of the tree (in terms of time slots) are not an option for these trees, even though they are commonly used for single channel trees \cite{massey1981collision}. On the contrary, a level-wise approach such as the one presented in \cite{kaplangulko} may still be applied, and it will be the basis of the analysis. For that, we first need to formally define the variables of the total and level-wise number of time slots.

\begin{definition}
Let $\mathcal{T}^N$ be the random variable modeling the number of time slots needed to complete a PCTA  given $N$ contenders in the root node.
\label{{def_T}}
\end{definition}
\begin{definition}
Let $\mathcal{T}^N_m$ be the random variable modeling the number of time slots within the level $m$. 
\label{{def_Tm}}
\end{definition}
From these two definitions, it follows that the total number of time slots in the PCTA, or just \textit{tree} for simplicity, can be expressed as:
\begin{equation}
\mathcal{T}^N = 1 + \sum_{m = 1}^\infty \mathcal{T}_m^N.
%\label{mc_avg_orig}
\label{avg_mc_s4}
\end{equation}
This equation is the starting point for calculating the statistical properties of $\mathcal{T}^N$.

\subsubsection{Probability mass function}
Our aim is to obtain the probability mass function (pmf) $p_{\mathcal{T}^N}(t)$ of the number of slots in the tree, provided the number of initial contenders $N$.
%\begin{equation}
% \triangleq \Pr\left\lbrace \mathcal{T}_N = t | n = N \right\rbrace.
%\end{equation}
%With the intention of reducing the complexity of the derivation of this function to an affordable level, simplifying approximations are presented.
%
%Even with approximations, the process of deriving the pmf of $\mathcal{T}_N$ is rather long and may be confusing. With the intention of facilitating understanding, there is a summary at the end of this section that compiles the most important steps of this process. This summary includes a table with the definition of the variables that are used throughout the analysis, which may be useful for the reader to consult throughout the analysis.
%\paragraph{Derivation approach and limitations}
In order to derive this pmf, we can use \eqref{avg_mc_s4} to express it as an infinite sum of the related random variables belonging to the following set.
\begin{definition}
	Let $\mathfrak{T}^N$ be the set of random variables $\mathcal{T}_m^N$ from $m=1$ to $m=\infty$:
	\begin{equation}
	\hat{\mathfrak{T}}^N_M \triangleq \left\lbrace \mathcal{T}_m^N : m \in \mathbb{N} \right\rbrace.
	\end{equation}
\end{definition}
With such an approach, we need to know the joint pmf of all variables in $\mathfrak{T}^N$, since it is clear from the properties of the tree that those variables are not independent from one another. However, a joint pmf of an infinite set of variables cannot be defined. Therefore, we have to limit our analyzed set of variables to a finite set such that the difference between the result yielded by the finite set and the actual result is negligible. 
With this in mind, we define a new, finite set of random variables with cardinality $M$.
\begin{definition}
	Let $\hat{\mathfrak{T}}^N_M$ be the set of random variables $\mathcal{T}_m^N$ from $m=1$ to $m=M$:
	\begin{equation}
	\hat{\mathfrak{T}}^N_M \triangleq \left\lbrace \mathcal{T}_m^N : m \in \mathbb{N} \; \wedge \; m \leq M \right\rbrace.
	\end{equation}
	The selection of $M$ and its effects on the accuracy of the result are discussed in App.~\ref{app:3}.
\end{definition}
We can now define the joint pmf of the variables in $\hat{\mathfrak{T}}^N_M$ as follows.
\begin{definition}
	Let $p_{\mathcal{T}^N_1,...,\mathcal{T}^N_M} (t_1,...,t_M)$ be the joint pmf of the variables in the set $\hat{\mathfrak{T}}^N_M$, that is:
	\begin{align}
		p_{\mathcal{T}^N_1,...,\mathcal{T}^N_M} (t_1,...,t_M) &\triangleq \Pr\left\lbrace \mathcal{T}^N_1 = t_1,...,\mathcal{T}^N_M = t_m \right\rbrace \\
		&= \Pr\left\lbrace  \left\langle\mathcal{T}^N_1,...,\mathcal{T}^N_M\right\rangle = \left\langle t_{1},...,t_{M}\right\rangle \right\rbrace. \label{prvector}
	\end{align}
	In \eqref{prvector}, a vectorial notation was used instead of the standard notation. This will be useful at some points in the subsequent analysis.
	\label{setor}
\end{definition}

All the statistical information of the number of slots in the tree is contained in the joint pmf of $p_{\mathcal{T}^N_1,...,\mathcal{T}^N_M} (t_1,...,t_M)$. Therefore, if this joint pmf was known, one could directly derive $p_{\mathcal{T}^N}(t)$. Indeed, these two pmfs are related as follows.

\begin{lemma}
	The probability $p_{\mathcal{T}^N}(t)$ can be expressed as a finite sum of values of the joint pmf of $\hat{\mathfrak{T}}^N_M$:
	\begin{equation}
	p_{\mathcal{T}^N}(t)
	= \sum_{\mathfrak{S}} p_{\mathcal{T}^N_1,...,\mathcal{T}^N_M} (t_1,...,t_M),
	\label{knoedel}
	\end{equation}
	where
	\begin{equation}
		 \mathfrak{S} = \left\lbrace \left\langle t_{1},...,t_{M}\right\rangle : \sum_{m=1}^M t_{m} = t-1 \right\rbrace  
	\end{equation}
	is the set of vectors of the possible realizations of the variables in the set $\hat{\mathfrak{T}}^N_M$ whose sum is $t-1$.
	\begin{proof}
		An element in $\mathfrak{S}$ is one distribution of level sizes (in time slots) such that the overall number of time slots in the tree is $t$. Hence, we just need to add the probabilities of all these combinations together ---which is given by the joint pmf of $\hat{\mathfrak{T}}^N_M$--- to obtain the probability of $\mathcal{T}^N = t$.
	\end{proof}
	\label{lemma_8}
\end{lemma}

%
%\begin{definition}
%We denote by $\mathcal{S}_N^A$ the sum of the first $A \leq M$ random variables in $\hat{\mathfrak{T}}^N$. That is:
%\begin{equation}
%\mathcal{S}_N^A \triangleq \sum_{m = 1}^A \mathcal{T}_N^m.
%\label{ajvar}
%\end{equation}
%\end{definition}
%Using this new definition, we can draw the following conclusion:
%\begin{lemma}
%	The pmf $p_{\mathcal{T}_N}(t)$ can be approximated by using the pmf of $\mathcal{S}_N^M$ when $M$ is large:
%	\begin{equation}
%	p_{\mathcal{T}_N}(t) \approx p_{\mathcal{S}_N^M}(t-1).
%	\label{taco}
%	\end{equation}
%\end{lemma}
%
%\begin{equation}
%\mathcal{T}_N \approx 1 + \mathcal{S}_N^M,
%\end{equation}
%
%Given the joint pmf of $\hat{T}_N$, the pmf of the variables in $\mathcal{S}_N^M$ can be expressed as:
%\begin{equation}
%p_{\mathcal{S}_N^M}(s) \triangleq \Pr\left\lbrace\mathcal{S}_N^M  = s \right\rbrace
%=\medmath{\sum_{t_1=0}^s \cdots \hspace{-0.4cm}\sum_{t_{M-1} = 0}^{s-\sum\limits_{i=1}^{M-2} t_m}\hspace{-0.4cm} p_{\mathcal{T}_N^1,\hdots,\mathcal{T}_N^M}\left(t_1,\hdots,t_{M-1},s-\textstyle\sum\limits_{i=1}^{M-1} t_i\right)}.
%\label{laleche}
%\end{equation}
%Although \eqref{laleche} may seem complicated, it just represents the addition of the probabilities of those combinations of  $\left\langle t_1,\hdots,t_M\right\rangle$ that yield the same sum. If we had this $p_{\mathcal{S}_N^M}(s)$, we could convert it directly to the pmf of $\mathcal{T}_N$, as follows:
%The larger $M$ (and hence $\epsilon$), the closer this approximation will be to the real value. 
The next step is to derive an expression for $p_{\mathcal{T}^N_1,...,\mathcal{T}^N_M} (t_1,...,t_M)$ as a function of $N$, since it will lead us to $p_{\mathcal{T}^N}(t)$. However, the derivation of this joint pmf is rather difficult, since we are facing the problem of finding out the relation among numerous variables that are all dependent from one another. In fact, attempting to model the exact dependence among all levels is likely to be cumbersome and even analytically intractable. Therefore, an approximative approach is presented.
%
%
%\begin{equation}
%p_{\mathcal{T}_m} (t_m | N) =  \begin{cases}
%p_{\mathcal{X}_{m-1}}(0 | N) & t_m = 0, \\
%\displaystyle \sum_{i = 1}^{G} p_{\mathcal{X}_{m-1}}\left(G \cdot (t_m-1)+i\right | N) & t_m > 0, \\
%\end{cases}
%\label{brezen}
%\end{equation}
%Furthermore, the conversion between conditional pmfs is obtained in a similar manner:
%
%
%Since $\mathcal{T}_m$ and $\mathcal{X}_m$ are directly related, we can obtain expressions equivalent to \eqref{lentejas} and \eqref{kartoffelsalat} for $\mathcal{X}_m$ instead of $\mathcal{T}_m$. Thus, our problem is now reduced to find $p_{\mathcal{X}_m | \mathcal{X}_{m-1}} (x_m | x_{m-1})$.
%\begin{equation}	
%	\label{cocoloco}
%\end{equation}
%This is the probability of getting $x_m$ collisions at the level $m$, given $x_{m-1}$ collisions at the level $m-1$. 
Namely, we use a Markovian approximation that exploits the level-by-level expanding nature of the trees. 

We will assume that the Markov property holds for our set of variables:
\begin{equation}
\Pr\left\lbrace\mathcal{T}^N_m = t_m | \mathcal{T}^N_{m-1} = t_{m-1}, \hdots, \mathcal{T}^N_{1} = t_{1}\right\rbrace
\cong \Pr\left\lbrace\mathcal{T}^N_m = t_m | \mathcal{T}^N_{m-1} = t_{m-1}\right\rbrace.
\label{lentejas}
\end{equation}
In words, this property implies that the number of time slots in a given level is only influenced by the number of time slots in the previous level. 

This assumption does not hold in general, since both the number of time slots and the distribution of contenders in those slots at the level $m-1$ are needed to calculate the statistics of the number of slots at the level $m$. The distribution of contenders is the result of what happened in the tree since the root node, which means that this information is not contained in the number of slots in the previous level. Therefore, $\mathcal{T}_m^N$ is indeed influenced by previous levels other than  $\mathcal{T}_{m-1}^N$. 

However, it is clear that the shorter the distance, the greater the dependence between two levels. Although it is not the only required information, the number of slots at the previous level is highly influential on the number of slots at any level. Thus, it is worth assuming that $\mathcal{T}^N_{m}$ only depends on $\mathcal{T}^N_{m-1}$, since such an approximation greatly simplifies the analysis and yet provides accurate results, as it will be shown through simulations. 

The first benefit of applying Markov property is the simple form of the joint pmf of the variables in $\hat{\mathfrak{T}}^N_M$, which is shown in the following lemma.
\begin{lemma} The joint pmf $p_{\mathcal{T}^N_1\!,\hdots,\mathcal{T}^N_M} (t_1,...,t_M)$ of the variables in $\hat{\mathfrak{T}}^N_M$ can be approximated as the product of the conditional pmfs of $\mathcal{T}_m$ and $\mathcal{T}_{m-1}$ for any two consecutive levels $m$ and $m-1$:
	\begin{equation}
	p_{\mathcal{T}^N_1\!,\hdots,\mathcal{T}^N_M} (t_1,...,t_M) \cong p_{\mathcal{T}^N_1}(t_1) \cdot  p_{\mathcal{T}^N_2 | \mathcal{T}^N_1}(t_2 | t_1) \cdot \hdots \cdot p_{\mathcal{T}^N_M | \mathcal{T}^N_{M-1}} (t_M | t_{M-1}) .
	\label{kartoffelsalat}
	\end{equation} 
\begin{proof}
	This is a well known property of Markov processes, in which the definition of conditional probability is combined with the Markov property:
	\begin{align}
	p_{\mathcal{T}^N_1\!,\hdots,\mathcal{T}^N_M} (t_1,...,t_M) &= p_{\mathcal{T}^N_M | \mathcal{T}^N_{M-1}\!,\hdots,\mathcal{T}^N_{1}}(t_M | t_{M-1},\hdots,t_1) \cdot p_{\mathcal{T}^N_1\!,\hdots,\mathcal{T}^N_{M-1}}(t_1,...,t_{M-1})\\
	&\cong p_{\mathcal{T}^N_M | \mathcal{T}^N_{M-1}}(t_M | t_{M-1}) \cdot p_{\mathcal{T}^N_1\!,\hdots,\mathcal{T}^N_{M-1}}(t_1,...,t_{M-1})
	\end{align}
	After iteratively applying the same procedure on $p_{\mathcal{T}^N_1\!,\hdots,\mathcal{T}^N_{M-1}}(t_1,...,t_{M-1})$ and onwards, we eventually reach \eqref{kartoffelsalat}.
\end{proof}
\label{lemma_approx}
\end{lemma}

With the result of Lemma \ref{lemma_approx} in mind, we can focus on the derivation of the conditional pmfs of the number of time slots at any level of the tree, provided the number of time slots at the previous level. We will tackle this problem by analyzing first the number of \textit{collisions} (the number of nodes with more than one contender) at each level. The number of collisions at a certain level can be easily translated into the number of time slots at the next level, as it will be shown in the next lemma. But before, we need to define a new variable to model the number of collisions.

\begin{definition}
	Let $\mathcal{X}^N_m$ be the random variable modeling the number of collisions within the level $m$, provided $N$ initial contenders. 
	\label{{def_Xm}}
\end{definition}

\begin{lemma}
	The conditional probability $p_{\mathcal{T}^N_m | \mathcal{T}^N_{m-1}} (t_m | t_{m-1})$ of obtaining $t_m$ time slots at the level $m$, provided $t_{m-1}$ time slots at the level $m-1$ can be expressed as:
	\begin{equation}
	\medmath{p_{\mathcal{T}^N_m | \mathcal{T}^N_{m-1}} (t_m | t_{m-1}) } \medmath{=} \!\begin{cases}
	\!\medmath{\sum_{i = 1}^{G} p_{\mathcal{X}^N_{m-1} | \mathcal{X}^N_{m-2}}\big(0\,\big|\,G \cdot (t_{m-1}-1)+i\big)} & \!\!\medmath{t_m = 0,} \\
	\!\medmath{\displaystyle \sum_{i = 1}^{G} \!\sum_{j = 1}^{G}\! p_{\mathcal{X}^N_{m\!-\!1} \!| \mathcal{X}^N_{m\!-\!2}} \!\big( G\! \cdot\! (t_m\!-\!1)\!+\!i \,\big|\,  G\! \cdot\! (t_{m\!-\!1}\!-\!1)\!+\!j\big)} & \!\!\medmath{t_m > 0,} \\
	\end{cases}
	\label{kaisersemmel}
	\end{equation}
	where $p_{\mathcal{X}^N_m | \mathcal{X}^N_{m-1}} (x_m | x_{m-1})$ is the conditional probability of obtaining $x_m$ collisions at the level $m$, provided $x_{m-1}$ collisions at the level $m-1$.
	\begin{proof}
		We know that every collision at the level $m-1$ produces two new nodes at the level $m$, and that one time slot contains $2 G$ nodes. Thus, we can convert collisions to time slots as follows:
		\begin{equation}
		\mathcal{T}^N_m = \left\lceil \frac{\mathcal{X}^N_{m-1}}{G} \right\rceil
		\label{lagrange}
		\end{equation}
		Owing to the presence of the ceiling function, the relation is not bijective, but several values of $\mathcal{X}^N_{m-1}$ map to the same value of $\mathcal{T}^N_m$. Indeed, given $\mathcal{T}^N_m = t_m$ and $\mathcal{X}^N_{m-1} = x_{m-1}$, any $x_{m-1}$ in the set $\left\lbrace G \cdot (t_m - 1) + i \;:\; 1 \leq i \leq G \right\rbrace$ fulfills \eqref{lagrange}. Hence, the conversion between the marginal probability $p_{\mathcal{T}^N_m} (t_m)$ of obtaining $t_m$ time slots at the level $m$ and the marginal probability $p_{\mathcal{X}^N_{m-1}} (x_{m-1})$ of obtaining $x_{m-1}$ collisions at the level $m-1$ is just a matter of adding together the probabilities of those $x_{m-1}$ that yield the same $t_{m}$:
		\begin{equation}
		p_{\mathcal{T}^N_m} (t_m) =  \begin{cases}
		p_{\mathcal{X}^N_{m-1}}(0) & t_m = 0, \\
		\displaystyle \sum_{i = 1}^{G} p_{\mathcal{X}^N_{m-1}}\left(G \cdot (t_m-1)+i\right) & t_m > 0, \\
		\end{cases}
		\label{brezen}
		\end{equation}
		Hence, in order to deduce the relation between $p_{\mathcal{T}_m | \mathcal{T}_{m-1}} (t_m | t_{m-1})$ and $p_{\mathcal{X}_m | \mathcal{X}_{m-1}} (x_m | x_{m-1})$, exactly the same procedure needs to be applied, but this time with two variables instead of one.
	\end{proof}
	\label{lemma_TTXX}
\end{lemma}

Provided Lemma \ref{lemma_TTXX}, the problem now is to find an expression for $p_{\mathcal{X}^N_m | \mathcal{X}^N_{m-1}} (x_m | x_{m-1})$ from the available information of the tree. In order to calculate this pmf, we need to know the number of contenders in each of the $x_{m-1}$ collisions of the level $m-1$. Indeed, the probability of producing, e.g., two new collisions is higher if the parent collision occurred with eight contenders than with four contenders. A priori, we cannot know the number of contenders involved in the given $x_{m-1}$ collisions. Instead we need to consider every different possibility and then apply the law of total probability. In order to do so, let us define a new variable for the number of contenders at each level.

\begin{definition}
	Let  $\mathcal{K}^N_m$ be the random variable modeling the total number of contenders which have been involved in collisions at the level $m$, i.e. the number of collided contenders at the level $m$, provided $N$ initial contenders.
	\label{{def_Km}}
\end{definition}
At this point, we are interested in the statistical properties of the distribution of contenders over nodes in the tree. Since contenders are treated in the same manner by the algorithm, and so are the nodes at one level, we can directly transform our \textit{contenders-into-nodes} problem into an equivalent \textit{balls-into-bins} problem. This simplifies the understanding of the problem and allows us to use existing solutions from the literature. 

The next two lemmas deal with the number of ways to distribute balls into bins such that some condition about the number or size of collisions is fulfilled. The results will be useful for subsequent lemmas.
\begin{lemma}
	The number of ways  $\Psi^N_{R,j}$ to arrange $N$ balls into $R$ bins such that $j \leq R$ of them have more than one ball can be obtained by means of the  recursion
	\begin{equation}
	\Psi_{R,j}^N = j \Psi_{R,j}^{N-1} + \left(R - j +1\right) \Psi_{R,j-1}^{N-1} + \Psi_{R-1,j}^{N-1},
	\label{caramba}
	\end{equation}
	with initial conditions $\Psi^N_{N,0} = 1$, $\Psi^1_{1,0} = 1$, $\Psi^1_{1,1} = 0$, and $\Psi^{N>1}_{1,1} = 1$.
	\begin{proof}
		The derivation of this recursion can be found in \cite{bianchi}.
	\end{proof}
	\label{lemma_bianchi}
\end{lemma}

\begin{lemma}
	Given an uniformly random distribution of $N$ balls into $R$ bins, the number of ways $\Gamma^{N,k}_{R,x}$ to generate  $x \leq R$ bins with more than one ball, such that the total number of balls occupying those $x$ bins is $k \leq N$, can be computed as:
	\begin{equation}
	\Gamma^{N,k}_{R,x} = \Psi^N_{x + N - k,x} {{R}\choose{x + N - k}} (x + N - k) !,
	\label{guthrie}
	\end{equation}
	where $\Psi^N_{R,j}$ was given in Lemma \ref{lemma_bianchi}.
	\begin{proof}
			Let us define $b_s$ and $b_t$ as the number of bins with only one ball and one or more balls, respectively, such that:
			\begin{equation}
			b_t = b_s + x.
			\end{equation}
			We have $N$ balls, $k$ of which are `contenders'. This implies that $N - k$ balls are alone in their occupied bin. Consequently:
			\begin{equation}
			b_s = N - k.
			\label{aladeuna}
			\end{equation}
			Therefore, we have a total of $x + N - k$ occupied bins, $x$ with collisions and $N - k$ with single balls:
			\begin{equation}
			b_t = x + N - k.
			\label{aladedos}
			\end{equation}
			Knowing this, we can compute the number of ways to arrange $N$ balls into $b_t$ bins such that $b_x$ of them have more than one ball, as given in Lemma \ref{lemma_bianchi}.  Finally, we just need to compute the number of ways to choose $b_t$ bins out of $R$ possible bins ---${{R}\choose{b_t}}$--- and the number of ways to arrange those bins --$b_t !$---. As a result, our final expression is:
			\begin{equation}
			\Gamma^{N,k}_{R,x} = \Psi^N_{b_t,b_x} {{R}\choose{b_t}} b_t !
			\label{aladetres}
			\end{equation}
			After combining \eqref{aladeuna}, \eqref{aladedos} and \eqref{aladetres}, we obtain \eqref{guthrie}. 
	\end{proof}
	\label{lemma_gamma}
\end{lemma}

We can now apply the result from Lemma \ref{lemma_gamma} to obtain the following intermediate probability, which will be employed in Lemma \ref{lemma_XX} to derive $p_{\mathcal{X}^N_m | \mathcal{X}^N_{m-1}} (x_m | x_{m-1})$.
\begin{lemma}
	The probability of having $\mathcal{K}^N_{m} = k_{m}$ collided contenders that formed $\mathcal{X}^N_{m} = x_{m}$ collisions at the level $m$ is
	
	\begin{equation}
	p_{\mathcal{K}^N_{m} | \mathcal{X}^N_{m}}(k_{m} | x_{m}) = \frac{\Gamma^{N,k_{m}}_{2^m,x_{m}}}{\sum\limits_{j = 0}^N \Gamma^{N,j}_{2^m,x_{m}}}.
	\label{muchas}
	\end{equation}
	\begin{proof}
		The number of ways to generate with exactly $k_{m}$ contenders from $x_{m}$ collisions is obtained by $\Gamma^{N,k_{m}}_{2^m,x_{m}}$, as given in Lemma \ref{lemma_gamma}. The number of bins in our case is $2^m$, since in each level of the tree the maximum number of nodes doubles, starting from 1 at level $m=0$. Finally, $\Gamma^{N,k_{m}}_{2^m,x_{m}}$ is divided by the total number of ways to produce $\mathcal{X}_{m} = x_{m}$ collisions for any possible value of $\mathcal{K}_{m}$, which is just the summation of $\Gamma^{N,j}_{2^m,x_{m}}$ from $j=0$ to $j = N$. 
	\end{proof}
	\label{lemma_KX}
\end{lemma}

The computation of the probability of generating $x_m$ collisions at the level $m$, given $x_{m-1}$ collisions and $k_{m-1}$ contenders at the level $m-1$ is another interesting statistic to investigate, since it will be also used to compute $p_{\mathcal{X}^N_m | \mathcal{X}^N_{m-1}} (x_m | x_{m-1})$  in Lemma \ref{lemma_XX}. The approach to obtain this probability relies on seeing the problem as a number theory problem. Namely, we will cope with integer partitions. The reason why this approach was chosen is illustrated in the following example.

\begin{example}
	Let us consider a scenario where $x_{m-1} =  4$ and $k_{m-1} = 12$. There are five different ways to decompose 12 contenders into 4 collisions, which are the five different partitions of 12 in 4 parts, such that every part is greater than one. Namely, these partitions are:
	\[\medmath{
		\begin{array}{ccc}
		(2,2,2,6) & \Rightarrow & 2+2+2+6=12, \\ 
		(2,2,3,5) & \Rightarrow & 2+2+3+5=12, \\ 
		(2,2,4,4) & \Rightarrow & 2+2+4+4=12, \\ 
		(2,3,3,4) & \Rightarrow & 2+3+3+4=12, \\ 
		(3,3,3,3) & \Rightarrow & 3+3+3+3=12.
		\end{array}} 
	\]
	
	At this point, it is easy to see why it is interesting to decompose $k_{m-1}$ into partitions. Given a certain partition of contenders at the level $m-1$, say $(2,2,3,5)$, it is immediate to compute the probability of $x_m$ collisions at the level $m$. We only need to compute the probability of generating 0, 1 or 2 new collisions for every existing collision (i.e., for every part of the partition), which is now simple since we know the number of contenders in each one.
	\label{ex_part}
\end{example}

We can incorporate partitions into the derivation of $p_{\mathcal{X}^N_m | \mathcal{X}^N_{m-1}} (x_m | x_{m-1})$ by using again the law of total probability. In order to do so, we first need to compute the probability of each partition to appear. 

%Nomenclature
%TODO
\begin{definition}
	Let $\mathfrak{P}^{k,x}$ be the set of partitions of $k$ in $x$ parts greater than 1. An element of $\mathfrak{P}^{k,x}$ is a partition $\pi^{k,s}_i$, such that:
	\begin{equation}
	\mathfrak{P}^{k,x} \triangleq \big\lbrace \pi^{k,x}_i :  i \in \left\lbrace 1, ..., \Pi(k,x)\right\rbrace \big\rbrace, \label{tomatina}
	\end{equation}
	where $\Pi(k,x)$ is the number of partitions of $k$ in $x$ parts greater than 1.
\end{definition}

\begin{definition}
	Let $\mathcal{P}^{k,x}$ be the random variable modeling one randomly chosen partition of $k$ contenders in $x$ parts greater than 1. That partition represents the distribution of the collided balls after an uniformly random allocation of $N$ balls into $R$ bins.
	\label{def_P}
\end{definition}
\begin{definition}
	  Let $p_{\mathcal{P}^{k,x}}(\pi^{k,x}_i)$ be the pmf of $\mathcal{P}^{k,x}$:
	  \begin{equation}
	  p_{\mathcal{P}^{k,x}}(\pi^{k,x}_i)  \triangleq \Pr\left\lbrace\mathcal{P}^{k,x} = \pi^{k,x}_i \;|\; \mathcal{X}_{m}^{N} = x, \mathcal{K}_{m}^{N}= k\right\rbrace, 
	  \label{bocapez}
	  \end{equation}
	  which represents the probability of ending up with the partition $\pi^{k,x}_i$ after a random arrangement of $k$ contenders into $x$ collisions at any level $m$. Note that the conditions in \eqref{bocapez} are implicit in the definition of $\mathcal{P}^{k,x}$, which allows us to simplify the notation.
\end{definition}
\begin{definition}
	 Let $\eta^{k,x}_{i,j}$ be a part of the partition $\pi^{k,x}_i$, for $j \in \{1, \hdots, x\}$. According to what was stated so far, the following relations hold:
	 \begin{equation}
	 \eta^{k,x}_{i,j} > 1
	 \end{equation}
	 \begin{equation}
	 \sum_{j=1}^{x} \eta^{k,x}_{i,j} = k,
	 \end{equation}
	 \begin{equation}
	 \pi^{k,x}_i = \left\langle\eta^{k,x}_{i,1}, \eta^{k,x}_{i,2}, \hdots, \eta^{k,x}_{i,x}\right\rangle.
	 \end{equation}
\end{definition}

\begin{definition}
	Let $\#^{k,x}_{i,a}$ be the number of occurrences of the number $a$ within the partition $\pi^{k,x}_i$. We can formally define this new variable as follows:
	\begin{equation}
	\#^{k,x}_{i,a} \triangleq \sum_{j=1}^{x} \left[\eta^{k,x}_{i,j} = a\right],
	\end{equation}
	where $\left[ \cdot \right]$ is the Iverson bracket, which returns 1 if the proposition inside is true. 
\end{definition}

With these definitions, we can compute the probability $p_{\mathcal{P}^{k,x}}(\pi^{k,x}_i)$ of encountering the partition $\pi^{k,x}_i$ as follows.
\begin{lemma}
	The probability to have the specific partitioning  $\pi^{k,x}_i$  of $k$ contenders in $x$ collisions is 
	\begin{equation}
	p_{\mathcal{P}^{k,x}}(\pi^{k,x}_i) = \frac{k!}{\Psi^{k}_{x,x}} \prod_{j=1}^{x} \frac{1}{\eta^{k,x}_{i,j}! \cdot \#^{k,x}_{i,a}!}.
	\label{weisswurst}
	\end{equation}
	\begin{proof}
		The derivation of \eqref{weisswurst} is explained in the Appendix \ref{LOL}.
	\end{proof}
		\label{lemma_part}
\end{lemma}

%\mg{Up to here connections were flowless. But should we rename the big joint condiitonal probabilities to better guide the reader with which one we are tackling with. I think here there is a jump.}
At this point, we have gathered all the information about the number and the distribution of the contenders that occupy some given collisions. This information almost suffices to compute the pmf of the number of collisions at the level $m$. The only missing part is an expression for computing the probability of generating a certain number of \textit{children} collisions, provided that we know the size of the \textit{parent} collision. We address the derivation of such an expression hereunder.
\begin{definition}
	Let $\mathcal{Y}_\eta$ be the random variable modeling number of child collisions of a parent collision of $\eta$ contenders. Since we are analyzing a Binary Tree Algorithm, the sample space of $\mathcal{Y}_\eta$ is simply $\{0,1,2\}$, i.e., at most two collisions can be children of one parent collision.
	\label{{def_Y}}
\end{definition}

\begin{lemma}
	The probability $p_{\mathcal{Y}_\eta}(y_\eta)$ of generating $y_\eta$ children collisions, provided a parent collision of size $\eta$ is:
	\begin{align}
	&p_{\mathcal{Y}_2}(y_2) = \begin{cases}
	\frac{1}{2} & y_2 = 0\\
	\frac{1}{2} & y_2 = 1\\
	0 & y_2 = 2
	\end{cases}
	&\text{If }\eta = 2,
	\\
	&p_{\mathcal{Y}_\eta}(y_\eta) = \begin{cases}
	0 & y_\eta = 0\\
	(\eta+1) \cdot \left( \frac{1}{2} \right)^{\eta-1} & y_\eta = 1\\
	1 - (\eta+1) \cdot \left( \frac{1}{2} \right)^{\eta-1} & y_\eta = 2
	\end{cases}
	&\text{If }\eta > 2.
	\end{align}
	\begin{proof}
			This pmf is the result of a simple combinatorial problem that can be decomposed in two different cases. If $\eta=2$, it is impossible two obtain 2 new collisions, and the events of 0 and 1 collisions are equally likely:
			\begin{equation}
			p_{\mathcal{Y}_2}(y_2) = \begin{cases}
			\frac{1}{2} & y_2 = 0,\\
			\frac{1}{2} & y_2 = 1,\\
			0 & y_2 = 2.
			\end{cases}
			\end{equation}
			On the other hand, if $\eta > 2$, it is impossible to generate 0 collisions, being the remaining options 1 or 2 collisions. In order to generate only 1 collision, one of the nodes needs to either be empty or contain a single contender. The probability of this situation can be computed by adding the probabilities of the following independent events: no contender chooses the first node, a single contender chooses the the first node, all contenders choose the first node, and all but one contenders choose the first node. Since choosing a node is a Bernoulli experiment with probability $\nu =\frac{1}{2}$, the probabilities of these events follow a binomial distribution:
			\begin{align}
				\Pr\left\{1 \text{ collision}\right\} &= \nu^\eta + \eta \nu (1 - \nu)^{\eta-1} + (1-\nu)^\eta + \eta\nu^{\eta-1} (1 - \nu)\\
				&= (\eta+1) \cdot \left( \frac{1}{2} \right)^{\eta-1}
			\end{align}
			Once we have computed the probability of 1 collision, the probability of 2 collisions is just the reciprocal:
			\begin{equation}
			p_{\mathcal{Y}_\eta}(y_\eta) = \begin{cases}
			0 & y_\eta = 0,\\
			(\eta+1) \cdot \left( \frac{1}{2} \right)^{\eta-1} & y_\eta = 1,\\
			1 - (\eta+1) \cdot \left( \frac{1}{2} \right)^{\eta-1} & y_\eta = 2.
			\end{cases}
			\end{equation}
	\end{proof}
\end{lemma}
As last step, we need to stitch together all the results that we have obtained in order to get a closed-form expression for $p_{\mathcal{X}^N_m | \mathcal{X}^N_{m-1}} (x_m | x_{m-1})$. The following three lemmas build upon the previous lemmas and yield such an expression.
\begin{lemma}
	The probability of having $x_{m}$ collisions at the level $m$, provided $x_{m-1}$ collisions and $k_{m-1}$ contenders partitioned in $\pi^{k_{m-1},x_{m-1}}_i$ at the level $m-1$ is
	\begin{multline}
p_{\mathcal{X}^N_m | \mathcal{X}^N_{m-1}, \mathcal{K}^N_{m-1}, \mathcal{P}^{k_{m-1},x_{m-1}}} (x_m | x_{m-1}, k_{m-1}, \pi^{k_{m-1},x_{m-1}}_i) = \\ = p_{\mathcal{Y}_{\eta^{k,x}_{i,1}}}\!\left(y_{\eta^{k,x}_{i,1}}\right) \ast \hdots \ast p_{\mathcal{Y}_{\eta^{k,x}_{i,x_{m-1}}}}\!\!\!\left(y_{\eta^{k,x}_{i,x_{m-1}}}\right),
\label{burrito}
	\end{multline}
	where $\ast$ denotes the discrete convolution.
	\begin{proof}
		Given a certain distribution (partition) of contenders, we can use $p_{\mathcal{Y}_\eta}(y_\eta)$ to compute the probability that some collision (part) at the level $m-1$ generates 0, 1 or 2 collisions at the level $m$. Furthermore, since the subtrees generated by the parent collisions are not related, variables $\mathcal{Y}_\eta$ are independent from one another. Therefore, we can compute the pmf of the sum of all $\mathcal{Y}_\eta$ as the discrete convolution of all of them.
	\end{proof}
	\label{lemma_3}
\end{lemma}
\begin{lemma}
	The probability to have $x_m$ collisions at the level $m$, given  $x_{m-1}$ collisions and $k_{m-1}$ contenders at the level $m-1$ is
	\begin{multline}
	p_{\mathcal{X}^N_m| \mathcal{X}^N_{m-1}, \mathcal{K}^N_{m-1}} (x_m | x_{m-1}, k_{m-1}) =\\
	= \sum_{\pi \in \mathfrak{P}^{x_{m-1}, k_{m-1}}} p_{\mathcal{X}^N_m| \mathcal{X}^N_{m-1}, \mathcal{K}^N_{m-1}, \mathcal{P}^{k_{m-1},x_{m-1}}} (x_m | x_{m-1}, k_{m-1}, \pi)  
	\cdot p_{\mathcal{P}^{k_{m-1},x_{m-1}}}(\pi).
	\label{senf}
	\end{multline}
	\begin{proof}
		This lemma is just an application of the law of total probability combining the expressions of  Lemma \ref{lemma_part} and Lemma \ref{lemma_3}.
	\end{proof}
	\label{lemma_4}
\end{lemma}
%Finally, after combining \eqref{brezen} -- \eqref{burrito}, we are able to compute the joint pmf of $\hat{T}_N$, which was defined in \eqref{kartoffelsalat}.
\begin{lemma}
	The probability to have $x_m$ collisions at the level $m$ given $x_{m-1}$ collisions at the level $m-1$ and $N$ initial contenders is
	\begin{equation}
	p_{\mathcal{X}^N_m | \mathcal{X}^N_{m-1}} (x_m | x_{m-1})  
	= \sum_{k_{m-1} = 0}^N p_{\mathcal{X}^N_m | \mathcal{X}^N_{m-1}, \mathcal{K}^N_{m-1}} (x_m | x_{m-1}, k_{m-1})    \cdot   p_{\mathcal{K}^N_{m-1} | \mathcal{X}^N_{m-1}}(k_{m-1} | x_{m-1}).
	\label{knoedel}
	\end{equation}
	\begin{proof}
		This is again a direct application of the law of total probability that combines the expressions of Lemma~\ref{lemma_4} and Lemma \ref{lemma_KX}.
	\end{proof}
	\label{lemma_XX}
\end{lemma}

%Therefore, at this point the problem is reduced to find:
%\begin{equation}
%p_{\mathcal{T}_m | \mathcal{T}_{m-1}} (t_m | t_{m-1}, N),
%\end{equation}
%which is the probability of obtaining $t_m$ slots at the level $m$ given $t_{m-1}$ at the level $m-1$. Since every collision at level $m-1$ produces a new pair of contention frame at level $m$, we can use the number of collisions at the level $m$, denoted by $\mathcal{X}_m(n)$ to obtain the number of time slots in a level.
%\begin{equation}
%
%\label{avg_mc_s3}
%\end{equation}
Finally, we have all the required ingredients to write down a close-form expression for the pmf of $\mathcal{T}^N$, which is shown in the following theorem.
\begin{theorem}
	The probability of tree successfully completing with $t$ time slots before level $M$ given $N$ initial contenders is 
	\begin{equation}
	p_{\mathcal{T}^N} (t)  
	\cong \sum_{\mathfrak{S}} p_{\mathcal{T}^N_1}(t_1) \prod_{m=2}^{M} p_{\mathcal{T}^N_m|\mathcal{T}^N_{m-1}} (t_m,t_{m-1})
	\label{knoedel}
	\end{equation}
	\begin{proof}
		Using the Markovian approximation we simplify the Lemma \ref{lemma_8}. Plugging the Lemmas \ref{lemma_approx} and \ref{lemma_TTXX}, which then use Lemma \ref{lemma_XX}, we are able to complete the calculation.
	\end{proof}
	\label{theorem_timeslot}
\end{theorem}

\subsection{Access delay or Probability of a contender to be successful in a time-slot}
\label{ramalamadindon_}
\begin{table}[]
	\centering
	\caption{Most relevant variables for computing $p_{\mathcal{T}^N}(t)$}
	\label{tableteble}
	\begin{tabular}{@{}clc@{}}
		\toprule
		\textbf{Variable} & \multicolumn{1}{c}{\textbf{Definition}}  &   \multicolumn{1}{c}{\textbf{Definition index}}                             \\ \midrule
		\vspace*{1mm} 
		$h$               & Level of successful transmission.    &            -                                                                      \\
		\vspace*{1mm}  
		$\mathcal{D}^N$   & Access delay in time slots.                  & \ref{def_D}                                                       \\
		\vspace*{1mm}
		$\mathcal{H}^N$   & Maximum level reached by the device.                  & \ref{def_H}                                                       \\
		\vspace*{1mm}
		$\mathcal{L}^N$   & Number of nodes in the tree.                  & \ref{def_L}                                                       \\
		\vspace*{1mm}
		$\mathcal{L}^N_h$   & Number of nodes at the level $h$.                  & \ref{def_Lh}                                                       \\
		\vspace*{1mm}
		$\mathcal{W}^N_h$   & Position of transmission in nodes at the level $h$.          & \ref{def_Wh}                                                       \\
		\vspace*{1mm}
		$\mathcal{V}^N_h$   & Position of transmission in time slots at the level $h$.          & \ref{def_Vh}                                               \\
		\vspace*{1mm}
		$\mathcal{S}^N_m$   & Number of time slots up to level $m$.                  & \ref{def_Sm}                                                       \\
		\vspace*{1mm}
		$\tilde{\mathcal{D}}^N$   & Access delay of a single channel tree.                 & \ref{def_Dt}                                                       \\
 \bottomrule
	\end{tabular}
\end{table}
In this section, we address the derivation of the statistics of the access delay of a multichannel tree. After the analysis of $\mathcal{T}^N$, we have many of the tools that we require in order to characterize this new variable. A summary of the new variables that will be used in this section can be found in Table \ref{tableteble}.

Since the access delay is usually more relevant than the length of the tree, it is worth addressing the calculation of its mean value separately. After this, the complete characterization of this random variable will be obtained by means of its pmf.

\subsubsection{Mean}
The first step to obtain a closed-form expression of the mean access delay is to formally define a random variable that models it. 

\begin{definition}
	Let $\mathcal{D}^N$ be the access delay of a focused, single device as the number of time slots between its first transmission and its successful transmission. Since the first time slot of a tree is always a collision, the access delay varies between 1 time slot (if the device transmits in the second time slot) and $\mathcal{T}^N-1$ (if the device transmits in the last time slot), where  $\mathcal{T}^N$ is the random variable modeling the number of time slots in the tree, which was studied in the last section.
	\label{def_D}
\end{definition}

In order to compute the average of $\mathcal{D}^N$, we introduce an auxiliary variable related with the number of levels tat a device traverses. The pmf of this variable is one of the requirements for the average of $\mathcal{D}^N$, as we will see later.
\begin{definition}
	Let $\mathcal{H}^N$ be the number of levels of the tree that the device has traversed until its successful transmission.
	\label{def_H}
\end{definition}
\begin{lemma}
	The pmf of the number of levels $\mathcal{H}^N$ of the tree that the device has traversed until its successful transmission is
	\begin{equation}
	p_{\mathcal{H}^N}(h)
	= \left(1-\frac{1}{2^h}\right)^{n-1} \!\!\!\!- \left(1-\frac{1}{2^{h-1}}\right)^{n-1}.
	\label{walkthisway_}
	\end{equation}
	\begin{proof}
		The proof of this lemma is given in \cite{janssenjong}.
	\end{proof}
	\label{lemma_H}
\end{lemma}

When referring to the level of successful transmission of the device, we will be using the index $h$ instead of the previously used $m$. This is to emphasize the special meaning of this level, since it is now linked to our focused device.  

We also need the probability of the number of collisions $\mathcal{X}^N_m$ at the arbitrary level $m$, since we can directly relate this random variable with the number of nodes ---and hence time slots--- at the next level. 
\begin{lemma}
	The pmf of the number of collisions $\mathcal{X}^N_m$ at the level $m$, provided $N$ initial devices is
		\begin{equation}
		p_{\mathcal{X}^N_{h}}(x_m) = \sum_{i = x_m}^{\min\left( N-x_m,Z \right)} \frac{\Psi_{i,x_m}^N {Z\choose{i}} i!}{Z^N},
		\label{lamia_}
		\end{equation}
		where $Z \triangleq 2^m$ is the maximum number of nodes at the level $m$.
		\begin{proof}
			The derivation of \eqref{lamia_} follows from the translation of a tree into a balls-into-bins problem, which is the same approach as that presented in \cite{bianchi}. We can again think of contenders as balls, and nodes as bins. Then, the probability that a contender is in a certain node at the level $m$ is equal to the probability that a randomly thrown ball ends up in a certain bin. Therefore, the $N$ contenders will be regarded as $N$ balls, and the level $m$ will be modeled as a set of $Z = 2^m$ bins. As a result, the probability of $x_m$ collisions at the level $m$ is equal to the probability of obtaining $x_m$ bins with more than one ball.
			
			We are then interested in computing the statistics of the number of bins with more than one ball after throwing $N$ balls into them. Initially, we have $Z$ bins, some of which might remain empty after the distribution of the $N$ balls. Therefore, we may consider that the balls have been divided in as many groups as non-empty bins. Nonetheless, the number of non-empty bins is also a random variable. This means that, in order to compute the probability of a certain number of collisions, we have to consider all different possibilities for the number of non-empty bins. Namely, the number of non-empty bins ranges from $x_m$ bins, which is the minimum number of bins such that they accommodate $x_m$ collisions, to either $Z$ ---the maximum number of nodes per level--- or $N-x_m$ bins, whichever is lower. The latter bound follows from the fact that the maximum number of \textit{successful} bins ---those that contain only one ball--- is $N - 2 x_m$, since at least two balls are needed in every collision. As the number of collisions has to be $x_m$, the total number of occupied bins cannot be greater than $N - 2 x_m + x_m = N -x_m$. These lower and upper bounds are the ones used by the summation in \eqref{lamia_}.
			
			Provided a number of occupied bins $i$, we need to find the number of ways to arrange $N$ balls into $x_m$ collisions. In order to do so, we benefit from an extension of the Stirling numbers of the second kind that we introduced in Lemma \ref{lemma_bianchi}, whose definition we repeat here for completeness. These numbers are denoted by $\Psi^N_{R,j}$, and are defined  as the number of ways to partition $N \geq 1$ balls into $R \geq 1$ subsets, such that $j \in (0, i)$ of them contain more than one ball. These numbers are defined in \cite{bianchi} by means of the recursion:
			\begin{equation}
			\Psi_{R,j}^N = j \Psi_{R,j}^{N-1} + \left(R - j +1\right) \Psi_{R,j-1}^{N-1} + \Psi_{R-1,j}^{N-1},
			\label{caramba_}
			\end{equation}
			and the initial conditions $\Psi^N_{N,0} = 1$, $\Psi^1_{1,0} = 1$, $\Psi^1_{1,1} = 0$, and $\Psi^{N>1}_{1,1} = 1$. 
			
			Along with this number, the binomial coefficient ${Z\choose{i}}$ and the factorial $i!$ in \eqref{lamia_} incorporate the number of ways to select $i$ non-empty bins out of $Z$ possible bins and the number of ways to arrange those $i$ bins, respectively. The last step to obtain the desired probability is to divide the number of ways to produce $x_m$ collisions by the total number of possible arrangements of $N$ balls into $Z$ bins, provided by $Z^N$.
		\end{proof}
\end{lemma}
	
	Another requirement for the average value of $\mathcal{D}^N$ is the pmf of the number of nodes contained in the level $h$. This is addressed by the followings definitions and lemma.
	
	\begin{definition}
		Let $\mathcal{L}^N$ be the random variable modeling the number of nodes contained in a tree, i.e., the length of the equivalent single channel tree.
		\label{def_L}
	\end{definition}
	\begin{definition}
		Let $\mathcal{L}^N_h$ be the random variable modeling the number of nodes within the level $h$ of a tree started with $N$ contenders, provided that the tree reaches such a level. That is, we assume that the level $h$ will not be empty, for reasons that will be obvious later on. 
		\label{def_Lh}
	\end{definition}
	 
\begin{lemma}
	The pmf $p_{\mathcal{L}^N_{h}}(l_h)$ of the number $\mathcal{L}^N_h$ of nodes within the level $h$, assuming that the tree extends at least up to that level, is
	\begin{equation}
	p_{\mathcal{L}^N_{h}}(l_h) = \begin{cases}
	\frac{p_{\mathcal{X}^N_{h-1}}\left(\frac{l_h}{2}\right)}{1-p_{\mathcal{X}^N_{h-1}}(0)} & \text{if $l_h$ even,} \\
	0 & \text{if $l_h$ odd.}
	\end{cases} \label{treat}
	\end{equation}
	\begin{proof}
		The probability of obtaining $l_h$ nodes at the level $h$ can be derived directly from the probability of obtaining half that number of collisions in the previous level, since every collision creates two new nodes. Moreover, since we are assuming that the successful transmission of our focused device occurred at the level $h$, we need to normalize the probability of a certain number of collisions to leave out the possibility of no collisions. Furthermore, it is clear that only even values of $l_h$ are allowed, due to the binary nature of the tree. 
	\end{proof}
\end{lemma}
Now that we have characterized the number of nodes at the level $h$ by means of $\mathcal{L}^N_h$, we are interested in knowing the statistics of the position of our successful transmission within that level. In order to do so, we introduce a new random variable, along with its pmf.
\begin{definition}
	Let $\mathcal{W}^N_h$ be the random variable modeling the number of nodes that lie between the first node belonging to the level $h$ and the node containing the successful transmission of the device, both inclusive. In other words, $\mathcal{W}^N_h$ reflects the position of the node of the successful transmission at the level $h$.
	\label{def_Wh}
\end{definition}

\begin{lemma}
	The pmf of $\mathcal{W}^N_h$ can be expressed as
		\begin{equation}
		p_{\mathcal{W}^N_h}(w_h) = \sum_{x = \left\lceil \!\frac{w_h}{2}\!\right\rceil}^{\infty} \frac{p_{\mathcal{X}^N_{h-1}}(x)}{2 x \left(1-p_{\mathcal{X}^N_{h-1}}(0)\right)}.
		\end{equation}
		\begin{proof}
			An unbiased binary tree is statistically symmetrical, i.e., it is equally probable for a device to transmit in any node of a given level. Hence it is clear that $\mathcal{W}^N_h$ follows an uniform distribution ranging from 1 to the total number of nodes $\mathcal{L}^N_h$ at the level $h$:
			\begin{equation}
			\mathcal{W}^N_h \sim U(1,\mathcal{L}^N_h).
			\end{equation}
			Since $\mathcal{L}^N_h$, the upper limit of $\mathcal{W}^N_h$, is itself a random variable, we need to apply the law of total probability to take into account every possible value.
			\begin{equation}
			p_{\mathcal{W}^N_h}(w_h) =  \sum_{l_h = 1}^{\infty}  \Pr \left\lbrace \mathcal{W}^N_h =w_h \, | \, \mathcal{L}^N_h = l_h\right\rbrace \cdot p_{\mathcal{L}^N_{h}}(l_h) \label{soulsacrifice_}
			\end{equation}
			As $\mathcal{W}^N_h$ is an uniform random variable, the probability of any value is constant, provided a deterministic upper bound. Therefore it follows that:
			\begin{equation}
			\Pr \left\lbrace \mathcal{W}^N_h =w_h \, | \, \mathcal{L}^N_h = l_h\right\rbrace = \begin{cases}
			\frac{1}{l_h} & \text{if } 1 \leq w_h \leq l_h,\\
			0 & \text{otherwise.}
			\end{cases} \label{jingo_}
			\end{equation}
			After combining \eqref{soulsacrifice_} and \eqref{jingo_}, the final expression for the pmf of $\mathcal{W}^N_h$ is obtained.
				\begin{equation}
				p_{\mathcal{W}^N_h}(w_h)= \sum_{l_h = w_h}^{\infty} \frac{p_{\mathcal{L}^N_{h}}(l_h)}{l_h}\\= \sum_{x = \left\lceil \!\frac{w_h}{2}\!\right\rceil}^{\infty} \frac{p_{\mathcal{X}^N_{h-1}}(x)}{2 x \left(1-p_{\mathcal{X}^N_{h-1}}(0)\right)}.
				\end{equation}
		\end{proof}
\end{lemma}
 
 So far, we have modeled the size of the level $h$ in terms of nodes, since they can be easily related with collisions. Nonetheless, the access delay has to be measured in time slots. Therefore, we need to convert the size of the level $h$ into time slots. With that intention in mind, we define the new variable  $\mathcal{V}^N_h$.
\begin{definition}
	Let $\mathcal{V}^N_h$ be the position of the time slot at the level $h$ in which the device successfully transmitted, assuming that the time slots within a certain level are numbered according to their transmission order. In other words, $\mathcal{V}^N_h$ is the number of time slots that are transmitted from the first slot belonging to the level $h$ until the time slot containing the successful transmission of our focused device, both inclusive.
	\label{def_Vh}
\end{definition}

\begin{lemma}
	\label{avgdel_lemma5}
	The mean of the number $\mathcal{V}^N_h$ of time slots that lie between the first time slot belonging to the level $h$ and the time slot containing the successful transmission of the device is
	\begin{equation}
	\mathrm{E}\left\lbrace\mathcal{V}^N_h\right\rbrace = \sum_{w_h=1}^{\min(N,2^h)} \!\!\! \left\lceil \frac{w_h}{2 G} \right\rceil \cdot p_{\mathcal{W}^N_h}(w_h).
	\label{ybicarbonato_}
	\end{equation}
	\begin{proof}
		 From the definition of $\mathcal{W}^N_h$, it follows that:
		 \begin{equation}
		 \mathcal{V}^N_h = \left\lceil \frac{\mathcal{W}^N_h}{2 G} \right\rceil.
		 \label{dexter}
		 \end{equation}
		 This relation is best deduced from an example. In Fig. \ref{fig:delay_explanation_}, $\mathcal{W}^N_h = 8$ nodes, since the successful transmission took place in the eighth node of level $h = 4$. Since $G = 2$, this implies that $\mathcal{V}^N_h = 2$, which is indeed the position of the time slot of the focused transmission.
		 
	 	\begin{figure}
	 		\centering
	 		\includegraphics[]{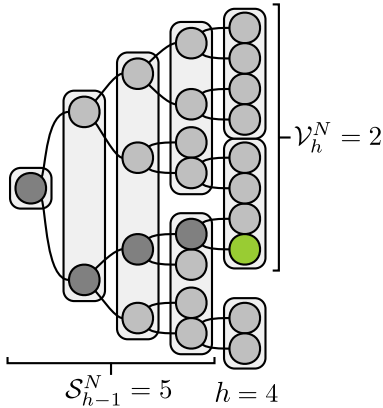}
	 		\caption{Example of tree depicting the variables $\mathcal{S}^N_{h-1}$ and $\mathcal{V}^N_h$, given a focused device (in green) that successfully transmits at the level $h=4$. Time slots (depicted as rectangles) follow $G =2$.}
	 		\label{fig:delay_explanation_}
	 	\end{figure}
		 
		 Finally, \eqref{ybicarbonato_} follows from the definition of the mean.
	\end{proof}
\end{lemma}

At this point, we have all the required elements to characterize the delay of our focused device within its level of successful transmission $h$. However, we miss a variable to model the delay caused by the time slots transmitted in previous levels. For that, we introduce a new variable.
\begin{definition}
We denote by $\mathcal{S}^N_m$ the sum of the first $m \leq M$ random variables in $\hat{\mathfrak{T}}^N_M$, as it was defined in Definition \ref{setor}. That is, the number of the time slots from level 1 to $m$:
\label{def_Sm}
\begin{equation}
\mathcal{S}^N_m \triangleq \sum_{i = 1}^m \mathcal{T}^N_i.
\label{ajvar}
\end{equation}
\end{definition}

With all these elements, we can obtain a closed-form expression for the mean access delay of a device in the MP-CTA, as follows.

\begin{theorem}
	The mean of the access delay $\mathcal{D}^N$ experienced by a device in a breadth-first multichannel tree (MP-CTA), provided $N$ initial devices, is
	\begin{equation}
\mathrm{E}\{\mathcal{D}^N\} = \sum_{h=1}^\infty \left(\sum_{k = 0}^{h-2} \sum_{x_k = 1}^{\hat{X}} \left\lceil \frac{x_k}{G} \right\rceil p_{\mathcal{X}^N_{k}}(x_k) \;+ \!\!\!\sum_{w_h=1}^{\min(N,2^h)} \!\!\left\lceil \frac{w_h}{2 G} \right\rceil \cdot p_{\mathcal{W}^N_h}(w_h)\right) \!\cdot p_{\mathcal{H}^N}(h). 
	\end{equation}
	where $\hat{X} = \min\left(\lfloor\frac{N}{2}\rfloor,2^m\right)$.
\begin{proof}
	By applying the law of total expectation, we can calculate the mean of $\mathcal{D}^N$ as:
	\begin{equation}
	\mathrm{E}\{\mathcal{D}^N\} = \sum_{h=1}^\infty\mathrm{E}\left\lbrace\mathcal{D}^N \, \big| \, \mathcal{H}^N = h\right\rbrace \cdot p_{\mathcal{H}^N}(h).
	\label{ificouldfly_}
	\end{equation}
	Notice that the summation above starts with $h=1$, since it is not possible to successfully transmit at the level 0. The only unknown in this expression is $\mathrm{E}\left\lbrace\mathcal{D}^N  | \mathcal{H}^N = h\right\rbrace$, which is the average delay experienced by a device provided that it has successfully transmitted at the level $h$. For a more convenient notation, let us introduce $\mathcal{D}^N_h$, defined as the random variable modeling the access delay of a device that has successfully transmitted at the level $h$:
	\begin{equation}
	\mathcal{D}^N_h \triangleq \mathcal{D}^N \big| _{\mathcal{H}^N = h}.
	\label{defidefi_}
	\end{equation}
	After using this new definition, \eqref{ificouldfly_} becomes:
	\begin{equation}
	\mathrm{E}\{\mathcal{D}^N\} = \sum_{h=1}^\infty\mathrm{E}\left\lbrace\mathcal{D}^N_h\right\rbrace \cdot p_{\mathcal{H}^N}(h).
	\label{crowdchant_}
	\end{equation}
	Owing to the structure of the tree, we can decompose $\mathcal{D}^N_h$ as follows:
	\begin{equation}
	\mathcal{D}^N_h = \begin{cases}
	1 & \text{if }h = 1,\\
	\mathcal{S}^N_{h-1} + \mathcal{V}^N_h  & \text{if }h > 1,
	\end{cases}
	\label{sambapati_}
	\end{equation}
	where $\mathcal{S}^N_{h-1}$ was defined in \eqref{ajvar} as the partial sum of the number of slots between levels $0$ and $h-1$.

	The meaning of this problem is illustrated in Fig. \ref{fig:delay_explanation_}, where our focused device has transmitted in the second slot of the fourth level, so that we can compute the access delay that it has experienced by counting the number of slots up to level 3 ($\mathcal{S}^N_{h-1} = 5$ slots) and then the number of slots within the level 4 ($\mathcal{V}^N_h = 2$ slots).
	
	If we apply the expectation operator in \eqref{sambapati_}, we obtain the following expressions for $h >1$:
	\begin{align}
	\mathrm{E}\{\mathcal{D}^N_h\} &= \mathrm{E}\left\lbrace\mathcal{S}^N_{h-1}\right\rbrace + \mathrm{E}\left\lbrace\mathcal{V}^N_h\right\rbrace\\
	&= \sum_{k = 1}^{h-1}\mathrm{E}\left\lbrace\mathcal{T}^N_{k}\right\rbrace + \mathrm{E}\left\lbrace\mathcal{V}^N_h\right\rbrace\\
	&= \sum_{k = 0}^{h-2}\mathrm{E}\left\lbrace\left\lceil\frac{\mathcal{X}^N_{k}}{G}\right\rceil\right\rbrace + \mathrm{E}\left\lbrace\mathcal{V}^N_h\right\rbrace\\
	&= \sum_{k = 0}^{h-2}\sum_{x_k = 1}^{\hat{X}} \left\lceil \frac{x_k}{G} \right\rceil p_{\mathcal{X}^N_{k}}(x_k) + \mathrm{E}\left\lbrace\mathcal{V}^N_h\right\rbrace
	\label{wellalright_}
	\end{align}
	At this point, we just need to replace in \eqref{wellalright_} the expression for $\mathrm{E}\left\lbrace\mathcal{V}^N_h\right\rbrace$ given in Lemma \ref{avgdel_lemma5} and then combine it with \eqref{crowdchant_}.
\end{proof}
\label{teorematopefino}
\end{theorem}
%%%%%%%%%%%%%%%%%%
%% NEXT SECTION %%
%%%%%%%%%%%%%%%%%%
\subsubsection{Probability mass function}
The goal of this section is to derive an expression for the pmf $p_{\mathcal{D}^N}(d)$ of $\mathcal{D}^N$. This function will yield a more insightful view of the access delay and will allow to provide stochastic guarantees to the devices using a MP-CTA. 

The first step is to relate the number of time slots at the level $h$ with the number of time slots up to that level. That is, we want to calculate the joint pmf of $\mathcal{T}^N_h$ and $\mathcal{S}^N_{h-1}$. This pmf is necessary since it relates the size of the level of successful transmission with the size of the rest of the tree, which is required to characterize the access delay. Although a closed-form expression is hard to obtain, we can rely on a recursive way of computing such a joint pmf, as it is shown in the lemma hereunder.
\begin{lemma}
	\label{pdel_lema2}
	The joint probability mass function $p_{\mathcal{T}^N_h, \mathcal{S}^N_{h-1}}(t_h,s_{h-1})$ of the number $\mathcal{T}^N_h$ of time slots at the level $h$ and the total number $\mathcal{S}^N_{h-1}$ of time slots up to the level $h-1$ can be computed recursively as
	\begin{equation}
	p_{\mathcal{T}^N_h, \mathcal{S}^N_{h-1}}(t_h,s_{h-1})
	\cong \sum_{t_{h-1} = 1}^{\hat{T}} p_{\mathcal{T}^N_h | \mathcal{T}^N_{h-1}}(t_h | t_{h-1})\cdot p_{\mathcal{T}^N_{h-1},\mathcal{S}^N_{h-2}}(t_{h-1},s_{h-1} - t_{h-1}),
	\label{currywurst_}
	\end{equation}
	where $\hat{T} = \left\lceil\frac{\min\left(2^{h-1}, N\right)}{2 G}\right\rceil$, with the initial conditions:
	\begin{equation}
	p_{\mathcal{S}^N_1}(s_1) = p_{\mathcal{T}^N_1}(s_1) = \delta_{s_1,0},
	\end{equation}
	\begin{equation}
	p_{\mathcal{T}^N_2, \mathcal{S}^N_1}(t_2, s_1) = p_{\mathcal{T}^N_2}(t_2) \cdot p_{\mathcal{S}^N_1}(s_1),
	\end{equation}
	where $\delta_{i,j}$ is the Kronecker delta.
	\begin{proof}
		In order to derive the recursion formula, we need to introduce the random variable $\mathcal{T}^N_{h-1}$, which can be accomplished by writing $p_{\mathcal{T}^N_h, \mathcal{S}^N_{h-1}}(t_h,s_{h-1})$ as a marginal pmf of the joint pmf of $\mathcal{T}^N_{h}$, $\mathcal{T}^N_{h-1}$ and $\mathcal{S}^N_{h-1}$, and then apply the law of total probability:
		\begin{align}
			p_{\mathcal{T}^N_h, \mathcal{S}^N_{h-1}}(t_h,s_{h-1}) 
			=& \sum_{t_{h-1} = 1}^{\hat{T}} p_{\mathcal{T}^N_h,  \mathcal{T}^N_{h-1},\mathcal{S}^N_{h-1}}(t_h, t_{h-1},s_{h-1})\\
			=& \sum_{t_{h-1} = 1}^{\hat{T}} p_{\mathcal{T}^N_h | \mathcal{T}^N_{h-1},\mathcal{S}^N_{h-1}}(t_h | t_{h-1},s_{h-1})\cdot p_{ \mathcal{T}^N_{h-1},\mathcal{S}^N_{h-1}}(t_{h-1},s_{h-1}) \\
			\cong& \sum_{t_{h-1} = 1}^{\hat{T}} p_{\mathcal{T}^N_h | \mathcal{T}^N_{h-1}}(t_h | t_{h-1})\cdot p_{\mathcal{T}^N_{h-1},\mathcal{S}^N_{h-2}}(t_{h-1},s_{h-1} - t_{h-1})
			\label{currywurst}
		\end{align}
		In the last step, the Markov property was applied, as well as the property $\mathcal{S}^N_{h-1} = \mathcal{S}^N_{h-2} + \mathcal{T}^N_{h-1}$. An expression for $p_{\mathcal{T}^N_h | \mathcal{T}^N_{h-1}}(t_h | t_{h-1})$ was already given in Lemma \ref{lemma_TTXX}.
		
		Regarding the initial conditions, it is clear that $\mathcal{S}^N_1 = 1$, since the first level always contains a single time slot. Due to this, $\mathcal{S}^N_1$ and $\mathcal{T}^N_2$ are independent, therefore their joint distribution can be written as the product of their marginal distributions.
	\end{proof}
\end{lemma}

Now, we need to find the relation between the position ---in time slots--- of our focused device within the level $h$, and the size ---in nodes--- of that level. This means that we need the pmf of $\mathcal{V}^N_h$ conditioned to a certain $\mathcal{L}^N_h$, which is given in the next lemma.
\begin{lemma}
	\label{antesmuertaquesencilla}
	The conditional probability $p_{\mathcal{V}^N_h | \mathcal{L}^N_h}(v_h | l_h)$ of the number $\mathcal{V}^N_h$ of time slots between the first and the successful transmission at the level $h$, provided the number $\mathcal{L}^N_h = l_h$ of nodes at the level $h$ is 
	\begin{gather}
	p_{\mathcal{V}^N_h | \mathcal{L}^N_h} (v_h | l_h) = \begin{cases}
	\frac{2 G}{l_h} & v_h < t_h, \\
	\frac{l_h - 2 G (t_h -1)}{l_h} & v_h = t_h,\\
	0 & v_h > t_h,
	\end{cases}
	\quad\text{ where  } 	t_h =  \left\lceil \frac{l_h}{2 G} \right\rceil.
	\label{cuerotroleo_}
	\end{gather}
	\begin{proof}
		This probability is simple to compute, since it is just the number of nodes that fit in one slot divided by the number $t_h$ of time slots in the level $h$. However, the number of nodes in a time slot is not fixed, but it ranges from $2$ to $2 G$ nodes. Namely, given $l_h$ nodes at the level $h$ that are grouped into $t_h$ slots, there will be $t_h -1$ slots of size $2 G$ and one (the last one) of size $l_h - 2 G (t_h-1)$. After taking into account this two cases, \eqref{cuerotroleo_} is reached.
	\end{proof}
\end{lemma}

As a last step, we need to link the two lemmas above, such that we can use them to compute the pmf of the access delay. With that intention in mind, we present the following lemma.
\begin{lemma}
	The conditional probability mass function $	p_{\mathcal{V}^N_h | \mathcal{T}^N_h, \mathcal{S}^N_{h-1}}(v_h | t_h, s_{h-1})$ of the number $\mathcal{V}^N_h$ of time slots between the first and the successful transmission at the level $h$, provided the number $\mathcal{T}^N_h = t_h$ of time slots at the level $h$, and the total number $\mathcal{S}^N_{h-1} = s_{h-1}$ of time slots up to the level $h-1$ is
	\begin{equation}
	p_{\mathcal{V}^N_h | \mathcal{T}^N_h, \mathcal{S}^N_{h-1}}(v_h | t_h, s_{h-1}) =
	\begin{cases}
		 \frac{1}{t_h} &  G = 1\\
		\displaystyle\sum_{l_h \in \mathfrak{L}_h} p_{\mathcal{V}^N_h | \mathcal{L}^N_h} (v_h | l_h) \cdot \frac{p_{\mathcal{L}^N_h}(l_h)}{p_{\mathcal{T}^N_h}\left(t_h \right)}\; &  G > 1,
	\end{cases}
	\label{ketchup}
	\end{equation}
	where
	\begin{align}
	\mathfrak{L}_h = \left\lbrace l_h \; : \; 2 (t_h-1) G +1\leq l_h \leq 2 t_h G\right\rbrace.
	\end{align}
	\begin{proof}
		If the tree is unbiased, as we have assumed for the whole analysis, all nodes within a tree are equally like, since the tree is statistically symmetrical. Therefore, if $G = 1$, every time slot is equally likely, since they all contain two nodes. Thus, the probability to transmit in a certain position given $t_h$ options is just $\frac{1}{t_h}$:
		\begin{equation}
		p_{\mathcal{V}^N_h | \mathcal{T}^N_h, \mathcal{S}^N_{h-1}} (v_h | t_h, s_{h-1}) =p_{\mathcal{V}^N_h | \mathcal{L}^N_h, \mathcal{S}^N_{h-1}} (v_h | 2 t_h, s_{h-1}) = \frac{1}{t_h}.
		\end{equation}
		However, in the case of $G > 1$, the number of nodes in one time slot may vary, as explained in Lemma \ref{antesmuertaquesencilla}. This difference in the size of the slots produces an asymmetry in the tree: it is less probable to transmit in the last slot than in the rest of them. In order to cope with this asymmetry, let us introduce the number of nodes $\mathcal{L}^N_h$ at the level $h$ into the problem, with the help of the law of total probability:
		\begin{equation}
		p_{\mathcal{V}^N_h | \mathcal{T}^N_h, \mathcal{S}^N_{h-1}} (v_h | t_h, s_{h-1}) = 
		\sum_{l_h \in \mathfrak{L}_h} p_{\mathcal{V}^N_h | \mathcal{L}^N_h, \mathcal{T}^N_h, \mathcal{S}^N_{h-1}} (v_h | l_h, t_h, s_{h-1}) \cdot p_{\mathcal{L}^N_h | \mathcal{T}^N_h, \mathcal{S}^N_{h-1}} (l_h | t_h, s_{h-1}).
		\label{grumete}
		\end{equation}
		where $\mathfrak{L}_h$ contains all the values of $l_h$ that are in agreement with the relation $t_h = \left\lceil \frac{l_h}{2 G} \right\rceil$, that is,
		\begin{equation}
		\mathfrak{L}_h = \left\lbrace l_h \; : \; t_h = \left\lceil \frac{l_h}{2 G} \right\rceil \right\rbrace = \left\lbrace l_h \; : \; 2 (t_h-1) G +1\leq l_h \leq 2 t_h G\right\rbrace.
		\end{equation}
		In the first term of \eqref{grumete}, the simplification 
		\begin{equation}
		p_{\mathcal{V}^N_h | \mathcal{L}^N_h, \mathcal{T}^N_h, \mathcal{S}^N_{h-1}} (v_h | l_h, t_h, s_{h-1}) = p_{\mathcal{V}^N_h | \mathcal{L}^N_h} (v_h | l_h)
		\end{equation}
		may be applied, since it is clear that the probability of transmitting in the slot $v_h$ is only influenced by the number of nodes $l_h$ at the level $h$. In the second term, the Bayes' theorem may be applied with the intention of changing the order of the variables. After such manipulations, the following expression is obtained:
		\begin{equation}
		p_{\mathcal{V}^N_h | \mathcal{T}^N_h, \mathcal{S}^N_{h-1}} (v_h | t_h, s_{h-1}) = \sum_{l_h \in L_h} p_{\mathcal{V}^N_h | \mathcal{L}^N_h} (v_h | l_h) \cdot p_{\mathcal{T}^N_h, \mathcal{S}^N_{h-1} | \mathcal{L}^N_h } (t_h, s_{h-1} | l_h)  \cdot \frac{p_{\mathcal{L}^N_h}(l_h)}{p_{\mathcal{T}^N_h, \mathcal{S}^N_{h-1}} (t_h, s_{h-1})}.
		\label{pineapple}
		\end{equation}
		Now we are left with the problem of calculating $p_{\mathcal{T}^N_h, \mathcal{S}^N_{h-1} | \mathcal{L}^N_h } (t_h, s_{h-1} | l_h)$, that is, the probability of obtaining $t_h$ slots at the level $h$ and a total of $s_{h-1}$ slots from level 1 to level $m-1$, given that we have $l_h$ nodes at the level $h$. We know that the number of slots is completely determined by the number of nodes, therefore we only allow for $t_h =  \left\lceil \frac{l_h}{2 G} \right\rceil$. This condition leads to the following expression:
		\begin{equation}
		p_{\mathcal{T}^N_h, \mathcal{S}^N_{h-1} | \mathcal{L}^N_h } (t_h, s_{h-1} | l_h) = \begin{cases}
		\frac{p_{\mathcal{T}^N_h, \mathcal{S}^N_{h-1}} (t_h, s_{h-1})}{ p_{\mathcal{T}^N_h}(t_h)} & t_h = \left\lceil \frac{l_h}{2 G} \right\rceil.\\
		0 & \text{otherwise.}
		\end{cases}
		\label{pepperoni_}
		\end{equation}
		After combining \eqref{pineapple} and \eqref{pepperoni_}, the final result is obtained:
		\begin{equation}
		p_{\mathcal{V}^N_h | \mathcal{T}^N_h, \mathcal{S}^N_{h-1}} (v_h | t_h, s_{h-1}) =  \sum_{l_h \in L_h} p_{\mathcal{V}^N_h | \mathcal{L}^N_h} (v_h | l_h) \cdot \frac{p_{\mathcal{L}^N_h}(l_h)}{p_{\mathcal{T}^N_h}\left(t_h \right)}.
		\end{equation}
	\end{proof}
	\label{lemma_VTS}
\end{lemma}

Finally, we have all the necessary elements to write down an expression for the probability mass function of the access delay.
\begin{theorem}
	\label{teorematopetocho}
	The probability mass function $p_{\mathcal{D}^N}(d)$ of the access delay $\mathcal{D}^N$ experienced by a device in a breadth-first multichannel tree is
	\begin{equation}
	p_{\mathcal{D}^N}(d) =
	\begin{cases}
	\displaystyle\sum_{h=1}^\infty \left( \sum_{v_{h}=1}^{\left\lceil\frac{2^h}{2 G}\right\rceil} \sum_{t_h} \frac{ p_{\mathcal{T}^N_h, \mathcal{S}^N_{h-1}} (t_h, d - v_{h})}{t_h}\right)\cdot p_{\mathcal{H}^N}(h) & G = 1\\
	\displaystyle \sum_{h=1}^\infty \left( \sum_{v_{h}=1}^{\left\lceil\frac{2^h}{2 G}\right\rceil} \sum_{l_h} \frac{p_{\mathcal{V}^N_h | \mathcal{L}^N_h} (v_h | l_h) \cdot p_{\mathcal{T}^N_h, \mathcal{S}^N_{h-1}} \left(\left\lceil\frac{l_h}{2 G}\right\rceil, d - v_{h}\right)  \cdot p_{\mathcal{L}^N_h}(l_h)}{p_{\mathcal{T}^N_h}\left(\left\lceil\frac{l_h}{2 G}\right\rceil \right)}\right)\cdot p_{\mathcal{H}^N}(h) & G > 1
	\end{cases}
	\end{equation}
\end{theorem}

\begin{proof}
	We can benefit from the law of total probability and the definition of $\mathcal{D}^N_h$ provided in Theorem \ref{teorematopefino} to express the pmf of $\mathcal{D}^N$ as follows:
	\begin{equation}
	p_{\mathcal{D}^N}(d) =\sum_{h=1}^\infty \Pr\left\lbrace\mathcal{D}^N_h = d \right\rbrace \cdot p_{\mathcal{H}^N}(h).
	\label{juan}
	\end{equation}
	Hence, we need to obtain the probability $p_{\mathcal{D}^N_h}(d_h)$ that our focused device successfully transmits with delay $d_h$, given the knowledge that it has transmitted at the level $h$.
	The variable $\mathcal{D}^N_h$ was defined in \eqref{sambapati_} as the sum of ${S}^N_{h-1}$ and  $\mathcal{V}^N_h$ for $h>1$. Again, we face the problem of obtaining the pmf of the sum of two dependent random variables. As in the previous section, this is accomplished by means of their joint pmf:
	\begin{equation}
	p_{\mathcal{D}^N_h}(d_h) = \sum_{v_{h}=1}^{\left\lceil\frac{2^h}{2 G}\right\rceil}  p_{\mathcal{V}^N_h, \mathcal{S}^N_{h-1}} (v_{h},d_h - v_{h}).
	\label{damian}
	\end{equation}
	The upper limit of the summation represents the maximum number of slots that can be obtained at the level $h$. At this point, the problem is to derive the joint pmf of $\mathcal{S}^N_{h-1}$ and $\mathcal{V}^N_h$. The law of total probability can be applied to introduce $\mathcal{T}^N_{h}$, which provides information about the number of slots at the level $h$. This is useful in order to infer the probability of transmitting at position $\mathcal{V}^N_h = v_h$.
	\begin{equation}
	p_{ \mathcal{V}^N_h, \mathcal{S}^N_{h-1}} (v_h, s_{h-1}) = \sum_{t_h} p_{\mathcal{V}^N_h | \mathcal{T}^N_h, \mathcal{S}^N_{h-1}} (v_h | t_h, s_{h-1}) \cdot p_{\mathcal{T}^N_h, \mathcal{S}^N_{h-1}} (t_h, s_{h-1}).
	\label{marron}
	\end{equation}
	
	The first term on the right-hand side of this equation was already derived as in Lemma \ref{lemma_VTS}, and the second term in Lemma \ref{pdel_lema2}. At this point, we just need to combine \eqref{juan}, \eqref{damian} and \eqref{marron} to obtain the final expression of the theorem. 
\end{proof}
Although all the previous analysis has focused on multichannel trees, we can also draw conclusions for single channel trees based on our results. These conclusions are presented as a corollary hereunder.
\begin{definition}
	Let  $\tilde{\mathcal{D}}^N$ be the access delay experienced by a device in a breadth-first single channel tree, provided $N$ initial contenders.
	\label{def_Dt}
\end{definition}
\begin{corollary}
	The probability mass function $p_{\tilde{\mathcal{D}}^N_h}(\tilde{d})$ of the access delay $\tilde{\mathcal{D}}^N$ experienced by a device in a breadth-first single channel tree is
	\begin{equation}
		p_{\tilde{\mathcal{D}}^N}(\tilde{d}) = \frac{1}{2}\sum_{h=1}^\infty \left( \sum_{v_{h}=1}^{\left\lceil\frac{2^h}{2 G}\right\rceil} \sum_{t_h} \frac{ p_{\mathcal{T}^N_h, \mathcal{S}^N_{h-1}} \left(t_h, \left\lceil \frac{\tilde{d}}{2} \right\rceil - v_{h}\right)}{t_h}\right)\cdot p_{\mathcal{H}^N}(h)
	\end{equation}
	\begin{proof}
	The access delay $\tilde{\mathcal{D}}^N$ experienced in a single channel tree and that experienced in a multichannel tree when $G = 1$ are directly related, since every time slot in the multichannel tree always contains two nodes. Namely, their relation is:
	\begin{equation}
		\mathcal{D}^N\Big|_{G=1} = \left\lceil\frac{\tilde{\mathcal{D}}^N}{2}\right\rceil
	\end{equation}
	Therefore, in order to compute the probability of $\tilde{\mathcal{D}}^N = \tilde{d}$, we can compute the probability of $\mathcal{D}^N = \left\lceil\frac{\tilde{d}}{2}\right\rceil$ for the equivalent multichannel tree with $G = 1$ by means of the theorem \ref{teorematopetocho}, and then divide that probability by two, since both nodes within each time slot are equally probable.
	\end{proof}
\end{corollary}

\section{Simulations}
\label{Simusimu}
In order to check the accuracy of the model, simulations were performed and their results were compared with the predicted values. The simulator was written in MATLAB, and the selected parameters were $N = \{10, 20, 30, 40, 50, 60\}$ contenders and $26500$ runs for each value of $N$.

\subsection{Simulation results}
In the Fig. \ref{fig:pmfs_delay_sim}, the theoretical and the empirical pmfs of the access delay for $G = \{1, ..., 4\}$ are plotted together for comparison. For $G = 1$, we see a slight but noticeable difference between the approximate model and the actual results, as a consequence of the approximate model. Nevertheless, this difference is rather small and the accuracy of the analytical model seems to improve rapidly when $G$ increases.
\begin{figure}
\centering
\includegraphics[width=1\linewidth]{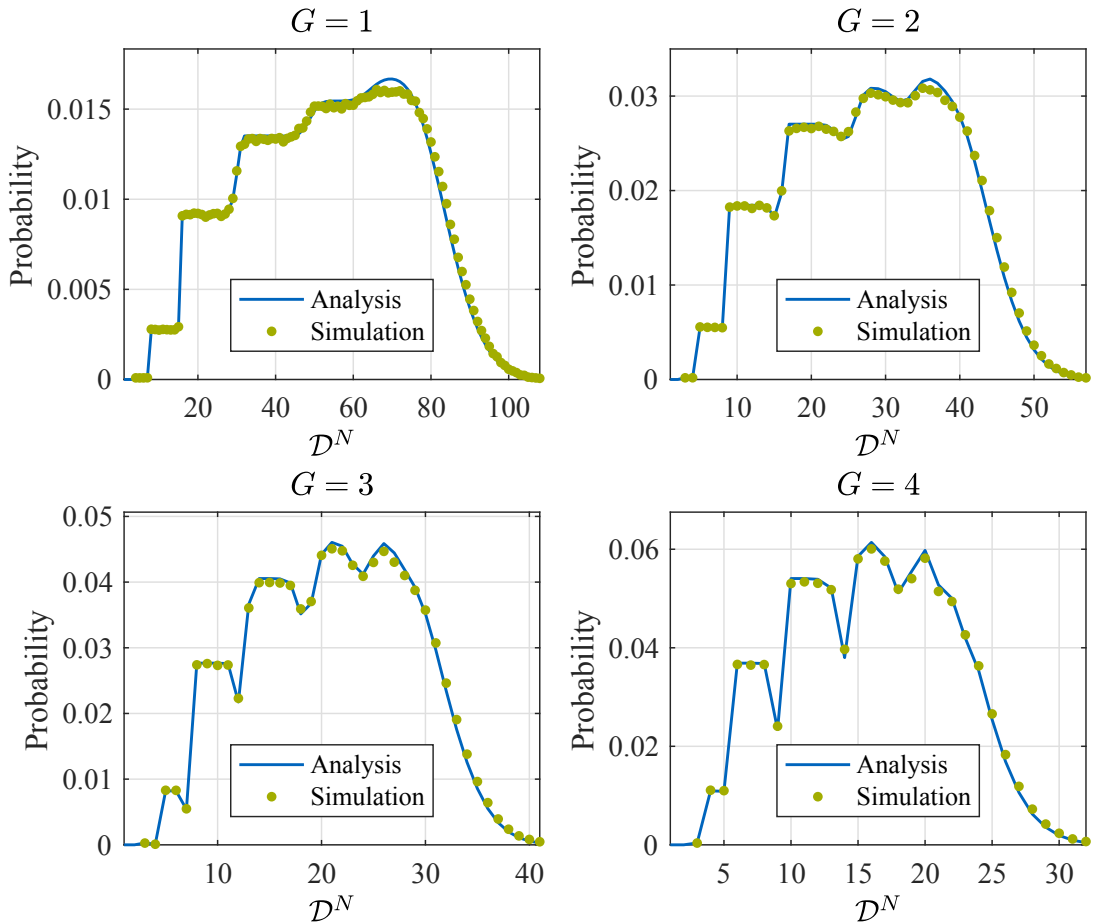}
\caption{Analytical and simulative results of the probability mass functions of the access delay experienced by a contender in a multichannel tree for several values of $G$, given $N=60$ initial contenders.}
\label{fig:pmfs_delay_sim}
\end{figure}

In the Fig. \ref{fig:cdfs_delay_sim}, the theoretical and the empirical CDFs of the access delay for $G = \{1, ..., 8\}$ are also plotted together. For $G = 1$, the predicted and the actual result differ slightly, although this is barely noticeable. For the remaining values of $G$, it can be observed that the model becomes more accurate when $G$ increases. Thus, one may conclude the Markovian approximation is valid and yields accurate approximations.

\begin{figure}
\centering
\includegraphics[width=1\linewidth]{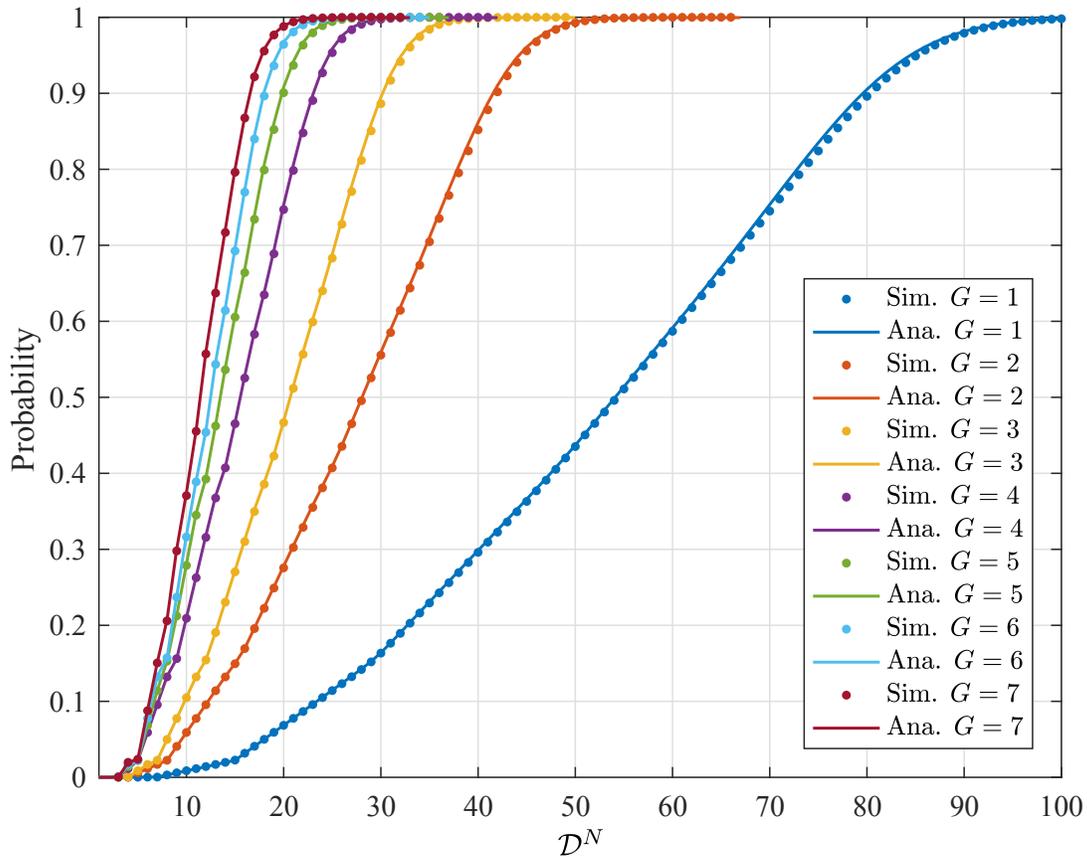}
\caption{Analytical and simulative results of the cumulative distribution functions of the access delay experienced by a contender in a multichannel tree for several values of $G$, given $N=60$ initial contenders.}
\label{fig:cdfs_delay_sim}
\end{figure}

Apart from the validity of the model, conclusions about the values of the access delay may be drawn as well, now that the predicted values are backed with simulations. Regarding access delay, we see how the maximum access delay for $G = 7$ might be lowered up to a 10\% of the delay of a single channel Tree Algorithm, which is obtained after multiplying by two the result for $G = 1$. Hence, a tenfold reduction of the access delay can be achieved if $G = 7$, and larger reductions are possible is $G \geq 7$. Nevertheless, the higher $G$ the lower the number of trees that can be executed in parallel if the number of channels is limited, therefore the optimum value of $G$ needs to be carefully chosen depending on the application.

\subsection{Goodness of the approximate model}
Although visual inspection of the the aforementioned figures seems to approve the validity the model, some measures are still necessary to be aware of the significance of the errors. In order to measure the goodness of fit between the approximate model and the simulation results, the Kolmogorov--Smirnov statistic \cite{kolmogorov} will be applied. This statistic is often employed to perform the Kolmogorov--Smirnov test, which is used to check whether an empirical CDF matches a theoretical CDF. Although this appears to be similar to our situation, it would be pointless to use the complete Kolmogorov--Smirnov test in the present case, since we already know that our theoretical CDF is just an approximation to the actual CDF. Therefore, the test will surely fail given a number of samples high enough. Nonetheless, the Kolmogorov--Smirnov statistic alone can be still used as a metric of the goodness of fit.

Given an empirical CDF of the access delay $\hat{F}^n_{\mathcal{D}^N}(d)$ and an analytical CDF $F_{\mathcal{D}^N}(d)$, the Kolmogorov--Smirnov statistic $KS$ is defined as:
\begin{equation}
KS = \max \left\vert \hat{F}^n_{\mathcal{D}^N}(d) - F_{\mathcal{D}^N}(d) \right\vert.
\end{equation}

In words, $D$ is simply the maximum difference between the empirical and the theoretical CDFs. As a consequence, in our case it will also be the maximum error that we could expect from our approximation. This statistic may seem biased, since only the worst point of the CDFs is considered, regardless of the goodness of the remaining points. However, as we are modeling an algorithm that might cope with delay-sensitive contenders, we are indeed interested in the maximum error of our prediction rather than in the average or some other `smoother' statistic. 

A table with values of $KS$ for different values of $G$ is presented in Table \ref{smirnov}, for the case of $N = 60$ contenders. One can observe from this table that the maximum difference between the predicted and the actual probability of any access delay is around 1\%, even lower for $G > 2$. For most applications, this margin of error should be acceptable. For other values of $N$, the evolution of $D$ with $G$ is depicted in Fig. \ref{fig:KS_stat_pic}. Although the behavior of this statistic is not smooth, it suggests that the accuracy of the model improves the higher $N$ and the lower $G$, but still the maximum difference is around 3.5\% when $G=1$ and $N=3$.
% \usepackage{booktabs}
\begin{table}[h]
	\centering
	\resizebox{0.5\linewidth}{!}{%
	\begin{tabular}{@{}ccccccccc@{}}
		\toprule
		$G$ & 1 & 2 & 3 & 4 & 5 & 6 & 7 \\ \midrule
		$KS$ & 0.0112 & 0.0103 & 0.0091 & 0.0085 & 0.0083 & 0.0060 & 0.0050 \\ \bottomrule
	\end{tabular}}
	\vspace{0.15cm}
	\caption{Kolmogorov--Smirnov statistic for different values of $G$ with $N$=60}
	\label{smirnov}
\end{table}
\begin{figure}
\centering
\includegraphics[width=0.7\linewidth]{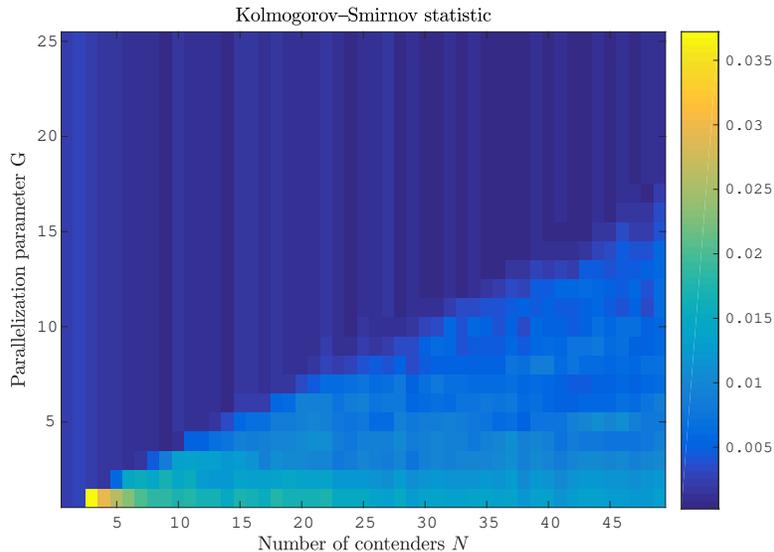}
\caption{Evolution of the Kolmogorov--Smirnov statistic with $G$ and $N$.}
\label{fig:KS_stat_pic}
\end{figure}

\section{Conclusion}
\label{sec:CONC}
In this work we provide detailed analysis and simulations of the statistics of multichannel Tree Algorithms. Namely, we derive the probability mass function of the length of a tree (in time slots), and the average and probability mass function of the access delay experienced by a contender, provided a number of initial contenders. We show that multichannel Tree Algorithms outperform single channel Tree Algorithms in terms of access delay. Owing to this property, multichannel Tree Algorithms can be used to deliver precise delays in those systems that are able to implement them, such as Ultra Reliable Low Latency Communications.

% or


\appendices

\section{Selection of $M$ for $\hat{\mathfrak{T}}^N_M$}
\label{app:3}
The set $\hat{\mathfrak{T}}^N_M$ contains the all the variables that will be used to compute an approximate pmf of $\mathcal{T}^N$. As a consequence, $M$ will be the maximum level that we will consider in the analysis of the pmf of $\mathcal{T}^N$. 

Since any node trespassing level $M$ will not be taken into account in the computation of the pmf, we want to set $M$ as high as possible. On the other hand, the greater $M$ the bulkier the operations will be, as more terms will be considered in them. In order to choose an optimum $M$, we need to compute the probability of a tree reaching the level $M$. With that objective in mind, let us define $\mathcal{M}^N$ as the random variable modeling the last level reached by a tree of $N$ contenders.
%, which is given by \eqref{paella}. 
In \cite{janssenjong}, the authors provide the pmf  of this random variable:
\begin{align}
p_{\mathcal{M}^N}(m) =\mu(2^m,N) - \mu(2^{m-1},N),
\label{paella}
\end{align}
where
\begin{equation}
\mu(\alpha,\beta) = \begin{cases}
0 & \text{if } \alpha < \beta,\\
\frac{\alpha !}{(\alpha - \beta)! \, \alpha^\beta} & \text{if } \alpha \geq \beta.
\end{cases}
\end{equation}
Provided that we have chosen an accuracy $\epsilon$, we need to select $M$ such that:
\begin{equation}
\epsilon \geq \sum_{m=1}^{M} \left(\mu(2^m,N) - \mu(2^{m-1},N)\right),
\end{equation}
which can be easily accomplished by numerical search. Then, we can choose the required $M$ for a desired accuracy. Fortunately, $M$ grows slowly as we increase either the required accuracy or the number of contenders, since the maximum number of nodes at each level grows exponentially, and so do the number of opportunities to successfully transmit. For instance, only $M=36$ is required to guarantee that at least $\epsilon=99.9$\% of trees will be finished even if $N = 10000$.

\section{Derivation of the probability of a given partition}
\label{LOL}
In \eqref{weisswurst}, the probability of obtaining a partition $\pi_i$ from a random distribution of $k_{h-1}$ balls over $x_{h-1}$ bins was presented as:
\begin{equation}
p_{\mathcal{P}^{k,x}}(\pi^{k,x}_i) = \frac{k!}{\Psi^{k}_{x,x}} \prod_{j=1}^{x} \frac{1}{\eta^{k,x}_{i,j}! \cdot \#^{k,x}_{i,a}!}.
\tag{\ref{weisswurst}}
\end{equation}

In this Appendix, the derivation of this expression will be tackled, using a slightly simplified notation for clearness. 
%First, let us simplify the notation, so that we denote the number of balls by $a$ and the number of bins by $b$. The probability of obtaining a partition $\pi$ of $a$ with $b$ parts is:
%\begin{equation}
%p_{\mathcal{P}^{a}_{b}}(\pi) = \frac{a!}{\Psi^{a}_{b,b}} \prod_{j=1}^{b} \frac{1}{\eta_j! \#_j!},
%\label{sardine}
%\end{equation}
%where $\pi = (\eta_1, \eta_2, \hdots, \eta_{b})$ and $\#_c$ is defined as:
%\begin{equation}
%\#_c = \sum_j \left[\eta_j = c\right],
%\end{equation}
%where $\left[ \cdot \right]$ is the Iverson bracket. In words, $\#^i_c$ is the number of occurrences of $c$ within the partition $\pi_i$. 
%
%Once all variables are defined, we can start with the derivation. 
Let us start by computing the number of ways $Z^{\eta_1}_k$ to choose $\eta_1$ balls out of a total of $k$ balls:
\begin{equation}
Z^{\eta_1}_k =  {k\choose{\eta_1}}.
\end{equation}
After this selection is done, the number of ways $Z^{\eta_2}_{k-\eta_1}$ to choose $\eta_2$ balls out of a total of $k-\eta_1$ balls is:
\begin{equation}
Z^{\eta_2}_{k - \eta_1} =  {{k-\eta_1}\choose{\eta_2}}.
\end{equation}
In general, the number of ways  $Z^{\eta_n}_{k}$ to choose $\eta_n$ balls out of a total of $k-\sum_{i = 1}^{n-1}\eta_i$ balls is:
\begin{equation}
Z^{\eta_n}_{k} =  {{a-\sum_{i = 1}^{n-1}\eta_i}\choose{\eta_n}}.
\end{equation}
If every part $\eta_n$ in the partition $\pi$ is different, the total number of ways $A^\pi$ to generate such partition is simply:
\begin{equation}
A^\pi = \prod_{j = 1}^x Z^{\eta_j}_k = \prod_{j = 1}^x {{k-\sum_{i = 1}^{j-1}\eta_i}\choose{\eta_j}}.
\label{frijoles}
\end{equation}
After some basic manipulation based on the definition of the binomial coefficient, we can rewrite \eqref{frijoles} as:
\begin{equation}
A^\pi = k! \prod_{j=1}^{x} \frac{1}{\eta_j!}.
\end{equation}
Nevertheless, if some parts have the same value, e.g. $10 = 4 + 4 + 2$, the number of ways to select those parts would be counted multiple times, yielding an incorrect result. In order to solve this issue, we have to correct by the number of ways to arrange those repeated values:
\begin{equation}
A^\pi = k! \prod_{j=1}^{x} \frac{1}{\eta_j!  \#_j!}.
\end{equation}
Finally, the probability of partition $\pi$ is obtained by dividing the number of ways $A^\pi$ to generate that specific partition by the total number of ways to generate any partition, which is given by $\Psi^{k}_{x,x}$, i.e. the number of ways to arrange $k$ balls in $x$ groups, $x$ of which have more than one ball. Therefore, we have reached our final result:
\begin{equation}
p_{\mathcal{P}^{k,x}}(\pi) = \frac{k!}{\Psi^{k}_{x,x}} \prod_{j=1}^{x} \frac{1}{\eta_j! \#_j!}
\label{mona}
\end{equation}
We just need to use the full notation in \eqref{mona} to obtain \eqref{weisswurst}.
\section*{Acknowledgment}

The authors would like to thank Markus Kl\"ugel for his useful comments.

\ifCLASSOPTIONcaptionsoff
  \newpage
\fi

\end{document}